\pdfoutput=1

\RequirePackage[l2tabu,orthodox]{nag}

\documentclass[11pt,letterpaper,reqno,oneside]{article}

\DeclareTextFontCommand{\emph}{\slshape}

\usepackage[T1]{fontenc}
\usepackage{lmodern}
\usepackage[utf8]{inputenc}

\usepackage[american,german,french,british]{babel}

\usepackage[activate={true,nocompatibility}]{microtype}

\usepackage{csquotes}

\usepackage[backend=biber,autolang=other,style=apa,maxcitenames=5,sortcites=false]{biblatex}
\DeclareLanguageMapping{british}{british-apa}
\DeclareLanguageMapping{american}{american-apa}
\DeclareLanguageMapping{german}{german-apa}
\DeclareLanguageMapping{french}{french-apa}
\DefineBibliographyExtras{french}{\restorecommand\mkbibnamefamily}

\usepackage{amsmath}
\usepackage{amssymb}
\usepackage{amsfonts}
\usepackage{amsthm}
\usepackage[retainorgcmds]{IEEEtrantools}
\usepackage{mathtools}
\usepackage{stmaryrd}
\usepackage{mleftright}
\usepackage{centernot}

\usepackage{graphicx}
\usepackage{tikz,pgfplots}
\pgfplotsset{compat=1.10}
\usetikzlibrary{shapes.multipart}
\usetikzlibrary{decorations.markings}
\usepackage{etoolbox}
\AtBeginEnvironment{tikzpicture}{\shorthandoff{;}}

\usepackage[title]{appendix}

\usepackage{datetime}
\newdateformat{datestyle}{ {\THEDAY} \monthname[\THEMONTH] \THEYEAR }

\usepackage{enumerate}
\usepackage{enumitem}
\setlist[enumerate,1]{label=(\arabic*)}
\setlist[itemize,1]{label=--}
\setlist[itemize,2]{label=--}
\setlist[itemize,3]{label=--}
\setlist[itemize,4]{label=--}

\usepackage{xcolor}
\usepackage{verbatim}
\usepackage{gensymb}
\usepackage{rotating}
\usepackage{subcaption}
\usepackage{floatpag}
\floatpagestyle{plain}
\usepackage{marvosym}

\usepackage{hyphenat}
\theoremstyle{definition}

\newtheorem{proposition}{Proposition}
\newtheorem{lemma}{Lemma}
\newtheorem{corollary}{Corollary}
\newtheorem{remark}{Remark}
\newtheorem{example}{Example}
\newtheorem{definition}{Definition}

\newtheoremstyle{named}
	{\topsep}					
	{\topsep}					
	{}							
	{0pt}						
	{\bfseries}					
	{}							
	{5pt plus 1pt minus 1pt}	
	{\thmnote{#3}}				
\theoremstyle{named}
\newtheorem{namedthm}{}

\usepackage{xpatch}
\xpatchcmd{\proof}{\itshape}{\proofheadfont}{}{}
\newcommand{\proofheadfont}{\slshape}

\usepackage{hyperref}
\hypersetup{
	colorlinks=true,
	linkcolor=black,
	citecolor=black,
	filecolor=black,
	urlcolor=black,
}

\usepackage[nameinlink]{cleveref}
\crefname{page}{p.}{pp.}
\crefname{equation}{equation}{equations}
\crefname{section}{section}{sections}
\crefname{subsection}{section}{sections}
\crefname{subsubsection}{section}{sections}
\crefname{appsec}{appendix}{appendices}
\crefname{supplsec}{supplemental appendix}{supplemental appendices}
\crefname{footnote}{footnote}{footnotes}
\crefname{figure}{figure}{figures}
\crefname{table}{table}{tables}
\crefname{theorem}{theorem}{theorems}
\crefname{proposition}{proposition}{propositions}
\crefname{lemma}{lemma}{lemmata}
\crefname{corollary}{corollary}{corollaries}
\crefname{remark}{remark}{remarks}
\crefname{observation}{observation}{observations}
\crefname{example}{example}{examples}
\crefname{fact}{fact}{facts}
\crefname{definition}{definition}{definitions}
\crefname{assumption}{assumption}{assumptions}
\crefname{exercise}{exercise}{exercises}
\crefname{notation}{notation}{notation}
\crefname{claim}{claim}{claims}
\crefname{conjecture}{conjecture}{conjectures}

\newcommand{\eps}{\varepsilon}

\newcommand{\dd}{\mathrm{d}}

\newcommand{\DD}{\mathrm{D}}

\DeclareMathOperator*{\interior}{int}

\newcommand{\E}{\mathbf{E}}

\newcommand{\oo}{\mathrm{o}}

\newcommand{\R}{\mathbf{R}}

\newcommand{\Q}{\mathbf{Q}}

\newcommand{\N}{\mathbf{N}}

\newcommand{\1}{\boldsymbol{1}}

\newcommand{\union}{\cup}

\newcommand{\intersect}{\cap}

\newcommand{\Union}{\bigcup}

\newcommand{\Intersect}{\bigcap}

\DeclarePairedDelimiter\abs{\lvert}{\rvert}

\DeclarePairedDelimiter\norm{\lVert}{\rVert}

\newcommand*{\xslant}[2][76]{%
	\begingroup
	\sbox0{#2}%
	\pgfmathsetlengthmacro\wdslant{\the\wd0 + cos(#1)*\the\wd0}%
	\leavevmode
	\hbox to \wdslant{\hss
		\tikz[
			baseline=(X.base),
			inner sep=0pt,
			transform canvas={xslant=cos(#1)},
		] \node (X) {\usebox0};%
		\hss
		\vrule width 0pt height\ht0 depth\dp0 %
	}%
	\endgroup
}
\makeatletter
\newcommand*{\xslantmath}{}
\def\xslantmath#1#{%
	\@xslantmath{#1}%
}
\newcommand*{\@xslantmath}[2]{%
	\ensuremath{%
		\mathpalette{\@@xslantmath{#1}}{#2}%
	}%
}
\newcommand*{\@@xslantmath}[3]{%
	\xslant#1{$#2#3\m@th$}%
}
\makeatother

\makeatletter
\let\save@mathaccent\mathaccent
\newcommand*\if@single[3]{%
	\setbox0\hbox{${\mathaccent"0362{#1}}^H$}%
	\setbox2\hbox{${\mathaccent"0362{\kern0pt#1}}^H$}%
	\ifdim\ht0=\ht2 #3\else #2\fi
	}
\newcommand*\rel@kern[1]{\kern#1\dimexpr\macc@kerna}
\newcommand*\widebar[1]{\@ifnextchar^{{\wide@bar{#1}{0}}}{\wide@bar{#1}{1}}}
\newcommand*\wide@bar[2]{\if@single{#1}{\wide@bar@{#1}{#2}{1}}{\wide@bar@{#1}{#2}{2}}}
\newcommand*\wide@bar@[3]{%
	\begingroup
	\def\mathaccent##1##2{%
	  \let\mathaccent\save@mathaccent
	  \if#32 \let\macc@nucleus\first@char \fi
	  \setbox\z@\hbox{$\macc@style{\macc@nucleus}_{}$}%
	  \setbox\tw@\hbox{$\macc@style{\macc@nucleus}{}_{}$}%
	  \dimen@\wd\tw@
	  \advance\dimen@-\wd\z@
	  \divide\dimen@ 3
	  \@tempdima\wd\tw@
	  \advance\@tempdima-\scriptspace
	  \divide\@tempdima 10
	  \advance\dimen@-\@tempdima
	  \ifdim\dimen@>\z@ \dimen@0pt\fi
	  \rel@kern{0.6}\kern-\dimen@
	  \if#31
	    \overline{\rel@kern{-0.6}\kern\dimen@\macc@nucleus\rel@kern{0.4}\kern\dimen@}%
	    \advance\dimen@0.4\dimexpr\macc@kerna
	    \let\final@kern#2%
	    \ifdim\dimen@<\z@ \let\final@kern1\fi
	    \if\final@kern1 \kern-\dimen@\fi
	  \else
	    \overline{\rel@kern{-0.6}\kern\dimen@#1}%
	  \fi
	}%
	\macc@depth\@ne
	\let\math@bgroup\@empty \let\math@egroup\macc@set@skewchar
	\mathsurround\z@ \frozen@everymath{\mathgroup\macc@group\relax}%
	\macc@set@skewchar\relax
	\let\mathaccentV\macc@nested@a
	\if#31
	  \macc@nested@a\relax111{#1}%
	\else
	  \def\gobble@till@marker##1\endmarker{}%
	  \futurelet\first@char\gobble@till@marker#1\endmarker
	  \ifcat\noexpand\first@char A\else
	    \def\first@char{}%
	  \fi
	  \macc@nested@a\relax111{\first@char}%
	\fi
	\endgroup
}
\makeatother

\title{\scshape The converse envelope theorem%
\thanks{I am grateful to Eddie Dekel,
Alessandro Pavan and
Bruno Strulovici
for their guidance and support.
This work has profited from the close reading and insightful comments of Gregorio Curello, Eddie Dekel, Roberto Saitto, Quitzé Valenzuela-Stookey, Alessandro Lizzeri
and four anonymous referees,
and from comments by
Piotr Dworczak, Matteo Escudé, Daniel Gottlieb, Elliot Lipnowski, Benny Moldovanu, Ilya Segal
and audiences at Caltech, Northwestern, Oxford, the Bonn Winter Theory Workshop, the Kansas Workshop in Economic Theory and the Southeast Theory Festival.}}

\author{Ludvig Sinander \\ University of Oxford}

\date{20 June 2022}

\makeatletter
	\AtBeginDocument{ \hypersetup{ 
		pdftitle = {The converse envelope theorem}, 
		pdfauthor = {Ludvig Sinander} 
		} }
\makeatother

\begin{document}

\maketitle

\begin{abstract}
	I prove an envelope theorem with a converse: the envelope formula is \emph{equivalent} to a first-order condition.
	Like Milgrom and Segal's (\citeyear{MilgromSegal2002}) envelope theorem, my result requires no structure on the choice set.
	I use the converse envelope theorem to extend to general outcomes and preferences the canonical result in mechanism design that any increasing allocation is implementable, and apply this to selling information.
\end{abstract}

\section{Introduction}
\label{sec:intro}

Envelope theorems are a key tool of economic theory, with important roles in consumer theory, mechanism design and dynamic optimisation.
In blueprint form, an envelope theorem gives conditions under which optimal decision-making implies that the \emph{envelope formula} holds.

In textbook accounts,%
	\footnote{E.g. \textcite[§M.L]{MascolellWhinstonGreen1995}.}
the envelope theorem is typically presented as a consequence of the first-order condition.
The modern envelope theorem of \textcite{MilgromSegal2002}, however, applies in an abstract setting in which the first-order condition is typically not even well-defined.
These authors therefore rejected the traditional intuition and developed a new one.

In this paper, I re-establish the intuitive link between the envelope formula and the first-order condition.
I introduce an appropriate generalised first-order condition that is well-defined in the abstract environment of \textcite{MilgromSegal2002},
then prove an envelope theorem with a converse: my generalised first-order condition is \emph{equivalent} to the envelope formula.
This validates the habitual interpretation of the envelope formula as `local optimality', and clarifies our understanding of the envelope theorem.

The converse envelope theorem proves useful for mechanism design.
I use it to establish that the implementability of all increasing allocations, a canonical result when outcomes are drawn from an interval of $\R$, remains valid when outcomes are abstract.
I apply this result to the problem of selling information (distributions of posteriors).

The setting is simple: an agent chooses an action $x$ from a set $\mathcal{X}$ to maximise $f(x,t)$, where $t \in [0,1]$ is a parameter.
The set $\mathcal{X}$ need not have any structure.
A \emph{decision rule} is a map $X : [0,1] \to \mathcal{X}$ that assigns an action $X(t)$ to each parameter $t$.
A decision rule $X$ is associated with a \emph{value function} $V_X(t) \coloneqq f(X(t),t)$, and is called \emph{optimal} iff $V_X(t) = \max_{ x \in \mathcal{X} } f(x,t)$ for every parameter $t$.

The modern envelope theorem of \textcite{MilgromSegal2002} states that, under a regularity assumption on $f$, any optimal decision rule $X$ induces an absolutely continuous value function $V_X$ which satisfies the \emph{envelope formula}
\begin{equation*}
	V_X'(t) = f_2( X(t), t )
	\quad \text{for a.e. $t \in (0,1)$.}
\end{equation*}
The familiar intuition is as follows. The derivative of the value $V_X$ is
\begin{equation*}
	V_X'(t) 
	= \left. \frac{ \dd }{ \dd m } f( X(t+m), t ) \right|_{m=0} 
	+ f_2( X(t), t ) ,
\end{equation*}
where the first term is the indirect effect via the induced change of the action, and the second term is the direct effect.
Since $X$ is optimal, it satisfies the first-order condition $\frac{ \dd }{ \dd m } f( X(t+m), t ) \big|_{m=0} = 0$, which yields the envelope formula.
Indeed, a decision rule $X$ satisfies the envelope formula \emph{if and only if} it satisfies the first-order condition for a.e. $t \in (0,1)$.

The trouble with this intuition is that since the action set $\mathcal{X}$ is abstract (with no linear or topological structure), the derivative $\frac{ \dd }{ \dd m } f( X(t+m), t ) \bigr|_{m=0}$ is ill-defined in general.

To restore the equivalence of the envelope formula and first-order condition, I first introduce a generalised first-order condition that is well-defined in the abstract environment. The \emph{outer first-order condition} is the following `integrated' variant of the classical first-order condition:
\begin{equation*}
	\left. \frac{\dd}{\dd m} \int_r^t f( X(s+m), s ) \dd s \right|_{m=0}
	= 0
	\quad \text{for all $r,t \in (0,1)$} .
\end{equation*}
I then prove an envelope theorem with a converse: under a regularity assumption on $f$, a decision rule $X$ satisfies the envelope formula \emph{if and only if} it satisfies the outer first-order condition and induces an absolutely continuous value function $V_X$.
The `only if' part is a novel \emph{converse} envelope theorem.

In §\ref{sec:app}, I apply the converse envelope theorem to mechanism design.
There is an agent with preferences over outcomes $y \in \mathcal{Y}$ and payments $p \in \R$.
Her preferences are indexed in `single-crossing' fashion by $t \in [0,1]$, and this taste parameter is privately known to her.
A canonical result is that if $\mathcal{Y}$ is an interval of $\R$, then all (and only) increasing allocations $Y : [0,1] \to \mathcal{Y}$ can be implemented incentive-compatibly by some payment schedule $P : [0,1] \to \R$.

I use the converse envelope theorem to extend this result
to a large class of ordered outcome spaces $\mathcal{Y}$,
maintaining general (non-quasi-linear) preferences.
The argument runs as follows: fix an increasing allocation $Y : [0,1] \to \mathcal{Y}$.
To implement it, choose a payment schedule $P : [0,1] \to \R$ to make the envelope formula hold.
Then by the converse envelope theorem, the outer first-order condition is satisfied, which means intuitively that $(Y,P)$ is \emph{locally} incentive-compatible.
The single-crossing property of preferences ensures that this translates into global incentive-compatibility.

I apply this implementability theorem to study the sale of information.
The result implies that any Blackwell-increasing information allocation is implementable.
I argue further that if consumers can share their information with each other, then \emph{only} Blackwell-increasing allocations are implementable.

\subsection{Related literature}
\label{sec:intro:lit}

Envelope theorems entered economics via the theories of the consumer and of the firm \parencite{Hotelling1932,Roy1947,Shephard1953}, were systematised by \textcite{Samuelson1947} under `classical' assumptions, and were developed in greater generality by e.g. \textcite{Danskin1966,Danskin1967}, \textcite{Silberberg1974} and \textcite{BenvenisteScheinkman1979}.
\textcite{MilgromSegal2002} pointed out that classical-type assumptions were extraneous, and proved an envelope theorem without them.
Subsequent refinements were obtained by e.g. \textcite{MorandReffettTarafdar2015} and \textcite{ClausenStrub2020}.%
	\footnote{See also \textcite{OyamaTakenawa2018}.}
`Converse' envelope theorems are almost absent from this literature, but appear in textbook presentations \parencite[e.g.][§M.L]{MascolellWhinstonGreen1995}.

The outer first-order condition appears to be novel. It bears no clear relationship to any of the standard derivatives for non-smooth functions.

\section{Setting and background}
\label{sec:background}

In this section, I introduce the environment, the Milgrom--Segal (\citeyear{MilgromSegal2002}) envelope theorem, and the classical envelope theorem and converse.

\begin{namedthm}[Notation.]
	\label{notation}
	We will be working with the unit interval $[0,1]$, equipped with the Lebesgue $\sigma$-algebra and the Lebesgue measure. The Lebesgue integral will be used throughout.
	For $r < t$ in $[0,1]$, we will write $\int_r^t$ for the integral over $[r,t]$, and $\int_t^r$ for $-\int_r^t$.
	$\mathcal{L}^1$ will denote the space of integrable functions $[0,1] \to \R$, i.e. those that are measurable and have finite integral.
	We will write $f_i$ for the derivative of a function $f$ with respect to its $i$th argument.
	Some important definitions and theorems are collected in \cref{app:theory:bckg},
	including \hyperref[theorem:LFTC]{Lebesgue's fundamental theorem of calculus}
	and the \hyperref[theorem:Vitali]{Vitali convergence theorem}.
\end{namedthm}

\subsection{Setting}
\label{sec:background:setting}

An agent chooses an action $x$ from an arbitrary set $\mathcal{X}$.
Her objective is $f(x,t)$, where $t \in [0,1]$ is a parameter (or `type').%
	\footnote{If instead the parameter lives in a normed vector space, then the analysis applies unchanged to path derivatives (as \textcite[footnote 7]{MilgromSegal2002} point out).}

\begin{definition}
	\label{definition:unif_abs_cont}
	A family $\{ \phi_x \}_{ x \in \mathcal{X} }$ of functions $[0,1] \to \R$ is \emph{absolutely equi-continuous} iff the family of functions
	\begin{equation*}
		\left\{ t \mapsto 
		\sup_{ x \in \mathcal{X} } \abs*{ \frac{ \phi_x(t+m) - \phi_x(t) }{m} }
		\right\}_{ m > 0 }
	\end{equation*}
	is uniformly integrable.%
		\footnote{The name `absolute equi-continuity' is inspired by the \hyperref[lemma:FitzpatrickHunt]{AC--UI lemma} in \cref{app:theory:bckg}, which states that absolute continuity of a continuous $\phi$ is equivalent to uniform integrability of the `divided-difference' family $\{ t \mapsto [ \phi(t+m) - \phi(t) ] / m \}_{ m > 0 }$.
		As the term suggests,
		an absolutely equi-continuous family
		is equi-continuous, and its members are absolutely continuous functions;
		this is proved in \cref{app:theory:AEC_HK}.}
\end{definition}

Our only assumptions will be that the objective varies smoothly, and (uniformly) not too erratically, with the parameter.

\begin{namedthm}[Basic assumptions.]
	\label{assumption:basic}
	$f(x,\cdot)$ is differentiable for every $x \in \mathcal{X}$, and the family $\{ f(x,\cdot) \}_{ x \in \mathcal{X} }$ is absolutely equi-continuous.
\end{namedthm}

\begin{remark}
	\label{remark:basic_suff}
	An easy-to-check sufficient condition for absolute equi-continuity is as follows: $f(x,\cdot)$ is absolutely continuous for each $x \in \mathcal{X}$, and there is an $\ell \in \mathcal{L}^1$ such that $\abs*{ f_2(x,t) } \leq \ell(t)$ for all $x \in \mathcal{X}$ and $t \in (0,1)$.
	(This is the assumption that \textcite{MilgromSegal2002} use in their envelope theorem.)
	An even stronger sufficient condition is that $f_2$ be bounded.
\end{remark}

\setcounter{example}{0}
\begin{example}
	\label{example:xt_assns}
	Let $\mathcal{X} = [0,1]$ and $f(x,t) = xt$.
	The \hyperref[assumption:basic]{basic assumptions} are satisfied since $f_2(x,t) = x$ exists and is bounded.
	\hfill \rotatebox[origin=c]{45}{$\square$}
\end{example}

A \emph{decision rule} is a map $X : [0,1] \to \mathcal{X}$ that prescribes an action for each type. The payoff of type $t$ from following decision rule $X$ is denoted $V_X(t) \coloneqq f( X(t), t )$.

\begin{definition}
	\label{definition:env}
	A decision rule $X$ satisfies the \emph{envelope formula} iff
	\begin{equation*}
		V_X(t) = V_X(0) + \int_0^t f_2( X(s), s ) \dd s
		\quad \text{for every $t \in [0,1]$} .
	\end{equation*}
\end{definition}

Equivalently (by \hyperref[theorem:LFTC]{Lebesgue's fundamental theorem of calculus}), $X$ satisfies the envelope formula iff $V_X$ is absolutely continuous and
\begin{equation*}
	V_X'(t) = f_2( X(t), t ) 
	\quad \text{for a.e. $t \in (0,1)$} .
\end{equation*}

A decision rule $X$ is called \emph{optimal} iff at every parameter $t \in [0,1]$, $X(t)$ maximises $f(\cdot,t)$ on $\mathcal{X}$.
The modern envelope theorem is as follows:
\begin{namedthm}[Milgrom--Segal envelope theorem.]
	\label{theorem:MS}
	Under the \hyperref[assumption:basic]{basic assumptions}, if $X$ is optimal, then it satisfies the envelope formula.
\end{namedthm}

This follows from the \hyperref[theorem:env]{main theorem} (§\ref{sec:result:mainthm} below), so no proof is necessary.
It is actually a slight refinement of Theorem 2 in \textcite{MilgromSegal2002}, as these authors impose the sufficient condition in \Cref{remark:basic_suff} rather than absolute equi-continuity.

\setcounter{example}{0}
\begin{example}[continued]
	\label{example:xt_envelope}
	The envelope formula requires that $X(t) t = \int_0^t X$ for every $t \in [0,1]$, or equivalently $X(t) = t^{-1} \int_0^t X$ for all $t \in (0,1]$.
	Thus the decision rules that satisfy the envelope formula are precisely those that are constant on $(0,1]$.
	This includes all optimal decision rules (which set $X=1$ on $(0,1]$), as well as anti-optimal ones (which choose $0$ on $(0,1]$).
	\hfill \rotatebox[origin=c]{45}{$\square$}
\end{example}

\subsection{Classical envelope theorem and converse}
\label{sec:background:classical}

The textbook version of the envelope theorem, which has a natural and intuitive converse, holds under additional topological and convexity assumptions.

\begin{namedthm}[Classical assumptions.]
	\label{assumption:classical}
	The action set $\mathcal{X}$ is a convex subset of $\R^n$, the action derivative $f_1$ exists and is bounded, and only Lipschitz continuous decision rules $X$ are considered.
\end{namedthm}

The \hyperref[assumption:classical]{classical assumptions} are strong. Most glaringly, the Lipschitz condition rules out important decision rules in many applications. In the canonical auction setting, for instance, the revenue-maximising mechanism is discontinuous \parencite{Myerson1981}.%
	\footnote{Even when the \hyperref[assumption:classical]{classical assumptions} are relaxed as much as possible, unless $f$ is trivial, $X$ still has to satisfy a strong continuity requirement. See \cref{app:theory:classical_env_pf}.}

\setcounter{example}{0}
\begin{example}[continued]
	\label{example:xt_classical_assns}
	$\mathcal{X} = [0,1]$ is a convex subset of $\R$, and $f_1(x,t) = t$ exists and is bounded.
	If we restrict attention to Lipschitz continuous decision rules $X : [0,1] \to [0,1]$, then the \hyperref[assumption:classical]{classical assumptions} are satisfied.
	\hfill \rotatebox[origin=c]{45}{$\square$}
\end{example}

Given a Lipschitz continuous decision rule $X$, suppose that type $t$ considers taking the action $X(t+m)$ intended for another type. The map $m \mapsto f(X(t+m),t)$ is differentiable a.e. under the \hyperref[assumption:classical]{classical assumptions},%
	\footnote{Since $f(\cdot,t)$ is differentiable, and $X$ is differentiable a.e. since it is Lipschitz continuous.}
so we may define a first-order condition:
\begin{definition}
	\label{definition:FOC_ae}
	A decision rule $X$ satisfies the \emph{first-order condition a.e.} iff
	\begin{equation*}
		\left. \frac{\dd}{\dd m} f( X(t+m), t ) \right|_{m=0}
		= 0
		\quad \text{for a.e. $t \in (0,1)$} .
	\end{equation*}
\end{definition}

The first-order condition a.e. requires that almost no type $t$ can secure a first-order payoff increase (or decrease) by choosing an action $X(t+m)$ intended for a nearby type $t+m$.
It does \emph{not} say that there are no nearby \emph{actions} that do better (or worse).

\begin{namedthm}[Classical envelope theorem and converse.]
	\label{theorem:classical_env}
	Under the \hyperref[assumption:basic]{basic} and \hyperref[assumption:classical]{classical} assumptions, a Lipschitz continuous decision rule satisfies the first-order condition a.e. iff it satisfies the envelope formula.
\end{namedthm}

The proof, given in \cref{app:theory:classical_env_pf}, shows that the envelope formula demands precisely that $V_X'(t) = f_2( X(t), t )$ for a.e. $t \in (0,1)$, 
which is equivalent to the first-order condition a.e. by inspection of the differentiation identity
\begin{equation*}
	V_X'(t) = \left. \frac{\dd}{\dd m} f( X(t+m), t ) \right|_{m=0} + f_2( X(t), t ) .
\end{equation*}

\setcounter{example}{0}
\begin{example}[continued]
	\label{example:xt_classical}
	A Lipschitz continuous decision rule $X$ is differentiable a.e., so satisfies the first-order condition a.e. iff
	\begin{equation*}
		\left. \frac{\dd}{\dd m} X(t+m) t \right|_{m=0} 
		= X'(t) t = 0
		\quad \text{for a.e. $t \in (0,1)$.}
	\end{equation*}
	This requires that $X'=0$ a.e.
	We saw that the envelope formula demands that $X$ be constant on $(0,1]$.
	For Lipschitz continuous decision rules $X$, both conditions are equivalent to constancy on all of $[0,1]$.
	\hfill \rotatebox[origin=c]{45}{$\square$}
\end{example}

\section{Main theorem}
\label{sec:result}

In this section, I define the outer first-order condition and state my envelope theorem and converse.

\subsection{The outer first-order condition}
\label{sec:result:oFOC}

Without the \hyperref[assumption:classical]{classical assumptions} (§\ref{sec:background:classical}), the `imitation derivative'
\begin{equation*}
	\left. \frac{\dd}{\dd m} f( X(t+m), t ) \right|_{m=0}
\end{equation*}
need not exist, in which case the first-order condition is ill-defined. To circumvent this problem, we require a novel first-order condition.

\begin{definition}
	\label{definition:oFOC}
	A decision rule $X$ satisfies the \emph{outer first-order condition} iff
	\begin{equation*}
		\left. \frac{\dd}{\dd m} \int_r^t f( X(s+m), s ) \dd s \right|_{m=0}
		= 0
		\quad \text{for all $r,t \in (0,1)$} .
	\end{equation*}
\end{definition}

As an intuitive motivation, suppose that types $s \in [r,t]$ deviate by choosing $X(s+m)$ rather than $X(s)$. The aggregate payoff to such a deviation is $\int_r^t f( X(s+m), s ) \dd s$, and the outer first-order condition says (loosely) that local deviations of this kind are collectively unprofitable.

\setcounter{example}{0}
\begin{example}[continued]
	\label{example:xt_oFOC}
	For any decision rule $X$ that is a.e. constant at some $k \in [0,1]$, the outer first-order condition holds:
	\begin{equation*}
		\left. \frac{\dd}{\dd m} \int_r^t X(s+m) s \dd s \right|_{m=0}
		= \left. \frac{\dd}{\dd m} k \int_r^t s \dd s \right|_{m=0}
		= 0
		\quad \text{for all $r,t \in (0,1)$} .
	\end{equation*}
	Conversely, any decision rule that is not constant a.e. violates the outer first-order condition.
	\hfill \rotatebox[origin=c]{45}{$\square$}
\end{example}

As we shall see, the outer first-order condition is well-defined even when the \hyperref[assumption:classical]{classical assumptions} fail. When they do hold, the outer first-order condition coincides with the first-order condition a.e.:
\begin{namedthm}[Housekeeping lemma.]
	\label{lemma:housekeeping}
	Under the \hyperref[assumption:basic]{basic} and \hyperref[assumption:classical]{classical} assumptions, the outer first-order condition is equivalent to the first-order condition a.e.
\end{namedthm}

\begin{proof}
	Fix a Lipschitz continuous decision rule $X : [0,1] \to \mathcal{X}$.
	The family
	\begin{equation*}
		\left\{ t \mapsto \frac{ f(X(t+m),t) - f(X(t),t) }{ m } \right\}_{ m > 0 }
	\end{equation*}
	is convergent a.e. as $m \downarrow 0$ by the \hyperref[assumption:classical]{classical assumptions},
	and is uniformly integrable by \Cref{lemma:classical} in \cref{app:theory:classical_lemma}.
	Hence by the \hyperref[theorem:Vitali]{Vitali convergence theorem}, for any $r,t \in (0,1)$,
	\begin{equation*}
		\left. \frac{\dd}{\dd m} \int_r^t f( X(s+m), s ) \dd s \right|_{m=0}
		= \left. \int_r^t \frac{\dd}{\dd m} f( X(s+m), s ) \right|_{m=0} \dd s .
	\end{equation*}
	The left-hand side (right-hand side) is zero for all $r,t \in (0,1)$ iff the outer first-order condition (first-order condition a.e.) holds.%
		\footnote{For the right-hand side, this relies on the following basic fact (e.g. Proposition 2.23(b) in \textcite{Folland1999}): for $\phi \in \mathcal{L}^1$, we have $\phi = 0$ a.e. iff $\int_r^t \phi = 0$ for all $r,t \in (0,1)$.}
\end{proof}

The term `outer' is inspired by this argument.
By taking the differentiation operator outside the integral, we change nothing in the classical case, and ensure existence beyond the classical case.

As its name suggests, the outer first-order condition is necessary (but not sufficient) for optimality. The following is proved in \cref{app:theory:necessity_pf}:

\begin{namedthm}[Necessity lemma.]
	\label{lemma:necessity}
	Under the \hyperref[assumption:basic]{basic assumptions}, 
	any optimal decision rule $X$
	satisfies the outer first-order condition,
	and has $V_X(t) \coloneqq f(X(t),t)$ absolutely continuous.
\end{namedthm}

\subsection{Envelope theorem and converse}
\label{sec:result:mainthm}

My main result characterises the envelope formula in terms of the outer first-order condition.

\begin{namedthm}[Envelope theorem and converse.]
	\label{theorem:env}
	Under the \hyperref[assumption:basic]{basic assumptions}, for a decision rule $X : [0,1] \to \mathcal{X}$, the following are equivalent:
	\begin{enumerate}
		
		\item \label{bullet:env_ofoc}
		$X$ satisfies the outer first-order condition
		\begin{equation*}
			\left. \frac{\dd}{\dd m} \int_r^t f( X(s+m), s ) \dd s \right|_{m=0}
			= 0
			\quad \text{for all $r,t \in (0,1)$} ,
		\end{equation*}
		and $V_X(t) \coloneqq f(X(t),t)$ is absolutely continuous.

		\item \label{bullet:env_env}
		$X$ satisfies the envelope formula
		\begin{equation*}
			V_X(t) = V_X(0) + \int_0^t f_2( X(s), s ) \dd s
			\quad \text{for every $t \in [0,1]$} .
		\end{equation*}
	
	\end{enumerate}
\end{namedthm}

The implication \ref{bullet:env_ofoc}$\implies$\ref{bullet:env_env} is an envelope theorem with weak (purely local) assumptions; the \hyperref[theorem:MS]{Milgrom--Segal} and \hyperref[theorem:classical_env]{classical} envelope theorems in §\ref{sec:background} are corollaries.
The implication \ref{bullet:env_env}$\implies$\ref{bullet:env_ofoc} is the converse envelope theorem, which entails the \hyperref[theorem:classical_env]{classical converse envelope theorem} in §\ref{sec:background:classical}.

The absolute-continuity-of-$V_X$ condition in \ref{bullet:env_ofoc} ensures that $f(X(\cdot),t)$ does not behave too erratically near $t$.
A characterisation of this property is provided in \cref{app:theory:AC-of-V_charac}.

\setcounter{example}{0}
\begin{example}[continued]
	\label{example:xt_thm}
	We saw that a decision rule satisfies the envelope formula iff it is constant on $(0,1]$ (p. \pageref{example:xt_envelope}), and satisfies the outer first-order condition iff it is constant a.e. (p. \pageref{example:xt_oFOC}).
	Thus the envelope formula implies the outer first-order condition.
	For the other direction, observe that an a.e. constant $X$ for which $V_X(t) = X(t) t$ is (absolutely) continuous must in fact be constant on $(0,1]$, though not necessarily at zero.
	\hfill \rotatebox[origin=c]{45}{$\square$}
\end{example}

In the classical case (§\ref{sec:background:classical}), our proof relied on the differentiation identity 
\begin{equation*}
	V_X'(t) = \left. \frac{\dd}{\dd m} f( X(t+m), t ) \right|_{m=0} + f_2( X(t), t ) ,
\end{equation*}
or (rearranged and integrated)
\begin{equation*}
	\int_r^t \left. \frac{\dd}{\dd m} f( X(s+m), s ) \right|_{m=0} \dd s 
	= V_X(t) - V_X(r) - \int_r^t f_2( X(s), s ) \dd s .
\end{equation*}
To pursue an analogous proof, we require an `outer' version of this identity in which differentiation and integration are interchanged on the left-hand side. The following lemma, proved in \cref{app:theory:pf_identity_lemma}, does the job.
\begin{namedthm}[Identity lemma.]
	\label{lemma:identity_lemma}
	Under the \hyperref[assumption:basic]{basic assumptions}, if $V_X$ is absolutely continuous, then for all $r,t \in (0,1)$,
	\begin{equation}
		\left. \frac{\dd}{\dd m} \int_r^t f( X(s+m), s ) \dd s \right|_{m=0}
		= V_X(t) - V_X(r) - \int_r^t f_2( X(s), s ) \dd s ,
		\label{eq:identity}
		\tag{$\mathcal{I}$}
	\end{equation}
	where both sides are well-defined.
\end{namedthm}

The left-hand side of \eqref{eq:identity} is zero for all $r,t \in (0,1)$ iff the outer first-order condition holds.
The right-hand side is zero for all $r,t \in (0,1)$ iff the envelope formula holds.%
	\footnote{For the `only if' part, if right-hand side is zero for all $r,t \in (0,1)$, then it is zero for all $r,t \in [0,1]$ since $V_X$ and the integral are continuous, yielding the envelope formula.}
Therefore:
\begin{proof}[Proof of the \protect{\hyperref[theorem:env]{envelope theorem and converse}}]
	Suppose that the outer first-order condition holds and that $V_X$ is absolutely continuous. Then the \hyperref[lemma:identity_lemma]{identity lemma} applies, so the outer first-order condition implies the envelope formula.

	Suppose that the envelope formula holds. Then $V_X$ is absolutely continuous by \hyperref[theorem:LFTC]{Lebesgue's fundamental theorem of calculus}. Hence the \hyperref[lemma:identity_lemma]{identity lemma} applies, so the envelope formula implies the outer first-order condition.
\end{proof}

\section{Application to mechanism design}
\label{sec:app}

A key result in mechanism design is that, 
provided the agent's preferences are `single-crossing', all and only increasing allocations are implementable.
While the `only' part is straightforward, the `all' part has substance.
Existing theorems of this sort require that outcomes be drawn from an interval of $\R$
or that the agent have quasi-linear preferences.

In this section, I use the \hyperref[theorem:env]{converse envelope theorem} to extend this result to abstract spaces of outcomes, without requiring quasi-linearity.
I then apply it to the problem of selling information, showing that all (and only) Blackwell-increasing information allocations are implementable (and robust to collusion).

\subsection{Environment and existing results}
\label{sec:app:env}

There is a partially ordered set $\mathcal{Y}$ of outcomes.
A single agent has preferences over outcomes $y \in \mathcal{Y}$ and payments $p \in \R$ represented by $f(y,p,t)$, where the type $t \in [0,1]$ is privately known to the agent.%
	\footnote{All of the analysis carries over to the case of multiple agents with independent types.}
We assume that $f( y, \cdot, t )$ is strictly decreasing and onto $\R$ for all $y \in \mathcal{Y}$ and $t \in [0,1]$.

A \emph{direct mechanism} is a pair of maps $Y : [0,1] \to \mathcal{Y}$ and $P : [0,1] \to \R$ that assign an outcome and a payment to each type. A direct mechanism $(Y,P)$ is called \emph{incentive-compatible} iff no type strictly prefers the outcome--payment pair designated for another type:
\begin{equation*}
	f( Y(t), P(t), t ) \geq f( Y(r), P(r), t )
	\quad \text{for all $r,t \in [0,1]$} .
\end{equation*}
By a revelation principle,
it is without loss of generality to restrict attention to incentive-compatible direct mechanisms. An allocation $Y : [0,1] \to \mathcal{Y}$ is called \emph{implementable} iff there is a payment schedule $P : [0,1] \to \R$ such that $(Y,P)$ is incentive-compatible.%
	\footnote{Adding an individual rationality constraint does not change our results below.}
An \emph{increasing} allocation is one that provides higher types with larger outcomes (in the partial order on $\mathcal{Y}$).

Preferences $f$ are called \emph{single-crossing} iff higher types are more willing to pay to increase $y \in \mathcal{Y}$.
The details of how this is formalised vary from paper to paper.
We are interested in the following type of result:

\begin{namedthm}[Theorem schema.]
	\label{namedthm:pseudo}
	If $\mathcal{Y}$ and $f$ are `regular' and $f$ is `single-crossing',
	then any increasing allocation is implementable.
\end{namedthm}

The first result of this kind was obtained by \textcite{Mirrlees1976} and \textcite{Spence1974} under the assumptions that $\mathcal{Y}$ is an interval of $\R$ and that $f$ has the quasi-linear form $f(y,p,t) = h(y,t) - p$.
Maintaining quasi-linearity, the result was extended to multi-dimensional Euclidean $\mathcal{Y}$ by \textcite{MatthewsMoore1987} and \textcite{Garcia2005},%
	\footnote{Results of this type have been used to study sequential screening (e.g. \textcite{CourtyLi2000}, \textcite{Battaglini2005}, \textcite{EsoSzentes2007}, and \textcite{PavanSegalToikka2014}).}
and may be further extended to arbitrary $\mathcal{Y}$
via a standard argument.
(That argument relies critically on quasi-linearity; see \cref{suppl:mech:standar_arg}.)
With $\mathcal{Y}$ an interval of $\R$,
the result was obtained without quasi-linearity by \textcite{GuesnerieLaffont1984} under \hyperref[assumption:classical]{classical assumptions},%
	\footnote{These authors restricted attention to piecewise continuously differentiable allocations; \textcite[Theorem 4.2]{Milgrom2004} generalised to piecewise absolutely continuous allocations.}
and by \textcite{NoldekeSamuelson2018} assuming only that $f$ is (jointly) continuous.

I shall extend the result to a wide class of outcome spaces $\mathcal{Y}$, without imposing quasi-linearity.
I formulate notions of `regularity' and `single-crossing' in the next section,
then establish the implementability of increasing allocations in §\ref{sec:app:incr_impl}.

\subsection{Regularity and single-crossing}
\label{sec:app:reg_sc}

Recall that a subset $\mathcal{C} \subseteq \mathcal{Y}$ is called a \emph{chain} iff it is totally ordered.

\begin{definition}
	\label{definition:outcome_reg}
	The outcome space $\mathcal{Y}$ is \emph{regular} iff it is order-dense-in-itself, countably chain-complete and chain-separable.%
		\footnote{A set $\mathcal{A}$ partially ordered by $\lesssim$ is \emph{order-dense-in-itself} iff for any $a<a'$ in $\mathcal{A}$, there is a $b \in \mathcal{A}$ such that $a<b<a'$.
		$B \subseteq \mathcal{A}$ is \emph{order-dense} in $C \subseteq \mathcal{A}$ iff for any $c<c'$ in $C$, there is a $b \in B$ such that $c \lesssim b \lesssim c'$.
		$\mathcal{A}$ is \emph{chain-separable} iff 	
		for each chain $C \subseteq \mathcal{A}$, there is a countable set $B \subseteq \mathcal{A}$ that is order-dense in $C$.
		$\mathcal{A}$ is \emph{countably chain-complete}
		iff every countable chain in $\mathcal{A}$ with a lower (upper) bound in $\mathcal{A}$ has an infimum (a supremum) in $\mathcal{A}$.}
\end{definition}

In words, $\mathcal{Y}$ must be `rich' (first two assumptions) and `not too large' (final assumption).
Many important spaces enjoy these properties, 
including $\R^n$ with the usual (product) order, 
the space of finite-expectation random variables (on some probability space) ordered by `a.s. smaller', and
the space of distributions of posteriors updated from a given prior ordered by Blackwell informativeness.
I prove these assertions and give further examples in \cref{suppl:mech:order_assns_ex}.

\begin{definition}
	\label{definition:pref_reg}
	The payoff $f$ is \emph{regular} iff
	(a) the type derivative $f_3$ exists and is bounded,
	and $f_3(y,\cdot,t)$ is continuous for each $y \in \mathcal{Y}$ and $t \in [0,1]$, and
	(b) for every chain $\mathcal{C} \subseteq \mathcal{Y}$,
	$f$ is jointly continuous on $\mathcal{C} \times \R \times [0,1]$ when $\mathcal{C}$ has the relative topology inherited from the order topology on $\mathcal{Y}$.%
		\footnote{The \emph{order topology} on $\mathcal{Y}$ is the one generated by the open order rays $\{ y' \in \mathcal{Y} : y' < y \}$ and $\{ y' \in \mathcal{Y} : y < y' \}$ for each $y \in \mathcal{Y}$, where $<$ denotes the strict part of the order on $\mathcal{Y}$.}$^,$%
		\footnote{It is sufficient, but unnecessarily strong, to assume joint continuity on $\mathcal{Y} \times \R \times [0,1]$.}
\end{definition}

The joint continuity requirement corresponds to Nöldeke and Samuelson's (\citeyear{NoldekeSamuelson2018}) regularity assumption.
By demanding in addition that the type derivative exist and be bounded,
I ensure that when this model is embedded in the general setting of §\ref{sec:background:setting} by letting $\mathcal{X} \coloneqq \mathcal{Y} \times \R$,
the \hyperref[assumption:basic]{basic assumptions} are satisfied.
The \hyperref[theorem:env]{converse envelope theorem} is thus applicable.%
	\footnote{The continuity of $f_3(y,\cdot,t)$ plays a technical role in the proof: see \cref{footnote:f3cont} below.}

It remains to formalise `single-crossing',
the idea that higher types are more willing to pay to increase $y \in \mathcal{Y}$.
Under the \hyperref[assumption:classical]{classical assumptions},
this is captured by the \emph{Spence--Mirrlees condition,}
which demands that for any increasing $Y : [0,1] \to \mathcal{Y}$ and any $P : [0,1] \to \R$ (both Lipschitz continuous), for any type $s \in (0,1)$,
the marginal gain to mimicking
\begin{equation*}
	\left. \frac{\dd}{\dd m} f(Y(s+m),P(s+m),s+n) \right|_{m=0} 
\end{equation*}
be single-crossing in $n$.%
	\footnote{\label{footnote:sc_defn}%
	Given $\mathcal{T} \subseteq \R$,
	a function $\phi : \mathcal{T} \to \R$
	is called \emph{single-crossing} iff for any $t < t'$ in $\mathcal{T}$, $\phi(t) \geq \mathrel{(>)} 0$ implies $\phi(t') \geq \mathrel{(>)} 0$,
	and \emph{strictly single-crossing} iff $\phi(t) \geq 0$ implies $\phi(t') > 0$.}$^,$%
	\footnote{An equivalent definition of the Spence--Mirrlees condition
	requires instead that
	the slope $f_1(y,p,t) / \abs*{ f_2(y,p,t) }$ of the agent's indifference curve through any point $(y,p) \in \mathcal{Y} \times \R$
	be increasing in $t$.
	See \textcite[Theorem 3]{MilgromShannon1994} for a proof of equivalence.}
To extend this definition
beyond the \hyperref[assumption:classical]{classical} case
to general outcomes $\mathcal{Y}$
(and non-Lipschitz mechanisms $(Y,P)$),
I replace the (typically ill-defined) marginal mimicking gain
with its `outer' version:

\begin{definition}
	\label{definition:SM}
	$f$ satisfies the \emph{(strict) outer Spence--Mirrlees condition} iff for any increasing $Y : [0,1] \to \mathcal{Y}$, any $P : [0,1] \to \R$
	and any $r<t$ in $(0,1)$,
	\begin{equation*}
		n \mapsto
		\left. \frac{\overline{\dd}}{\overline{\dd} m} \int_r^t
		f(Y(s+m),P(s+m),s+n) \dd s
		\right|_{m=0}
	\end{equation*}
	is (strictly) single-crossing,
	where $\overline{\dd} / \overline{\dd} m$ denotes the upper derivative.%
		\footnote{The \emph{upper derivative} of $\phi : [0,1] \to \R$ at $t \in (0,1)$ is $\smash{\frac{\overline{\dd}}{\overline{\dd} m} \phi(t+m) 
		\bigr|_{m=0}}
		\coloneqq \limsup_{ m \to 0 } \left[ \phi(t+m) - \phi(t) \right] / m$.
		Nothing changes in the sequel
		if the upper derivative is replaced with the lower (defined with a $\liminf$),
		or with any of the four Dini derivatives.}
\end{definition}

The difference from the classical Spence--Mirrlees condition is merely technical:
the interpretation is the same, viz. that on the margin, higher types have a greater willingness to pay for increasing the outcome $y \in \mathcal{Y}$.
It is worth noting, however, that whereas the classical Spence--Mirrlees condition is (nearly) ordinal,%
	\footnote{Precisely: if $f$ satisfies this condition, then so does $\phi \circ f$ for any differentiable and strictly increasing transformation $\phi : \R \to \R$.}
the outer Spence--Mirrlees condition is not.

\subsection{Increasing allocations are implementable}
\label{sec:app:incr_impl}

\begin{namedthm}[Implementability theorem.]
	\label{theorem:incr_impl}
	If \hyperref[definition:outcome_reg]{$\mathcal{Y}$} and \hyperref[definition:pref_reg]{$f$} are regular
	and $f$ satisfies the \hyperref[definition:SM]{outer Spence--Mirrlees condition},
	then any increasing allocation is implementable.
\end{namedthm}

The proof is in \cref{app:mech:pf_incr_impl}.
The idea is as follows.
Take any increasing allocation $Y : [0,1] \to \mathcal{Y}$.
By the \hyperref[lemma:existence]{existence lemma} in \cref{app:mech:pf_incr_impl:solns_env},
there exists a payment schedule $P : [0,1] \to \R$ such that $(Y,P)$ satisfies the envelope formula.%
	\footnote{\label{footnote:f3cont}This is where the continuity of $f_3(y,\cdot,t)$ is used: the \hyperref[lemma:existence]{existence lemma} requires it.}
By the \hyperref[theorem:env]{converse envelope theorem}, it follows that $(Y,P)$ is locally incentive-compatible in the sense that it satisfies the outer first-order condition.
The \hyperref[definition:SM]{outer Spence--Mirrlees condition} ensures that local incentive-compatibility translates into global incentive-compatibility.

The argument for the final step actually applies only to allocations $Y$ that are suitably continuous.
But the \hyperref[definition:outcome_reg]{regularity of $\mathcal{Y}$} ensures (via a lemma in \cref{app:mech:pf_incr_impl:approx}) that any increasing $Y$ can be approximated by a sequence of continuous and increasing (hence implementable) allocations.

Given two mild additional assumptions, the payment rule implementing a given increasing allocation is in fact unique, and may be computed constructively via Picard's method---see \cref{app:mech:pf_incr_impl:solns_env}.

The \hyperref[theorem:incr_impl]{implementability theorem} admits a standard converse when $\mathcal{Y}$ is a chain (e.g. an interval of $\R$),
proved in \cref{app:mech:impl_incr}:

\begin{proposition}
	\label{proposition:impl_equiv_incr}
	If \hyperref[definition:outcome_reg]{$\mathcal{Y}$} and \hyperref[definition:pref_reg]{$f$} are regular,
	$f$ satisfies the \hyperref[definition:SM]{strict outer Spence--Mirrlees condition},
	and $\mathcal{Y}$ is a chain,
	then all and only increasing allocations are implementable.
\end{proposition}

\subsection{Selling information}
\label{sec:app:info}

In this section, I apply the \hyperref[theorem:incr_impl]{implementability theorem} to selling informative signals.
Here the outcomes $\mathcal{Y}$ are distributions of posterior beliefs---a space very different from an interval of $\R$.
I show that all Blackwell-increasing information allocations are implementable,
and that only these are implementable if agents are able to share information with each other.

There is a population of agents with types $t \in [0,1]$, a finite set $\Omega$ of states of the world, and a set $A$ of actions.
A type-$t$ agent earns payoff $U(a,\omega,t)$ if she takes action $a \in A$ in state $\omega \in \Omega$, so her expected value at belief $\mu \in \Delta(\Omega)$ is
\begin{equation*}
	V(\mu,t) 
	\coloneqq \sup_{ a \in A } \sum_{ \omega \in \Omega }
	U(a,\omega,t) \mu(\omega) .
\end{equation*}
Assume that the type derivative $V_2$ exists and is bounded, and that $V_2(\cdot,t)$ is continuous for each $t \in [0,1]$.%
	\footnote{This is slightly stronger than assuming that the underlying type derivative $U_3$ has the same properties; see e.g. \textcite[Theorem 3]{MilgromSegal2002} for sufficient conditions.}

\begin{example}
	\label{example:forecasting}
	Each agent is tasked with announcing a probabilistic forecast $a \in A \coloneqq \Delta(\Omega)$ of the state $\omega \in \Omega$.
	Ex post, the public's assessment of an agent's quality as a forecaster is some function of the forecast $a$ and realised state $\omega$ (a \emph{scoring rule}); for concreteness, $a(\omega) / \norm*{a}_2$, where $\norm*{\cdot}_2$ denotes the Euclidean norm.%
		\footnote{More generally, any bounded and strictly proper scoring rule will do. See e.g. \textcite{GneitingRaftery2007} for an introduction to proper scoring rules.}
	Each agent attaches some importance $t \in [0,1]$ to being considered a good forecaster, so that $U(a,\omega,t) = t a(\omega) / \norm{a}_2$.
	Agents are expected-utility maximisers.

	It is easily verified that an agent with belief $\mu \in \Delta(\Omega)$ optimally announces forecast $a=\mu$.
	Her value is therefore
	\begin{equation*}
		V(\mu,t) 
		= \sum_{\omega \in \Omega} 
		\frac{ t \mu(\omega) }{ \norm*{\mu}_2 } \mu(\omega)
		= t \norm*{\mu}_2 .
	\end{equation*}
	By inspection, $V_2(\mu,t) = \norm*{\mu}_2$ exists, is bounded, and is continuous in $\mu$.
	\hfill \rotatebox[origin=c]{45}{$\square$}
\end{example}

Agents share a common prior $\mu_0 \in \interior \Delta(\Omega)$.
Before making her decision, an agent observes the realisation of a signal
(a random variable correlated with $\omega$),
and forms a posterior belief according to Bayes's rule.
Since the signal is random, the agent's posterior is random; write $y$ for its distribution (a Borel probability measure on $\Delta(\Omega)$).
The agent's expected payoff under a signal that induces posterior distribution $y$, if she makes payment $p \in \R$, is
\begin{equation*}
	f(y,p,t) \coloneqq g\left( \int_{\Delta(\Omega)} V(\mu,t) y( \dd \mu ), p \right) ,
\end{equation*}
where $g : \R^2 \to \R$ is jointly continuous,
possesses a bounded derivative $g_1$ that is continuous in $p$,
and has $g(v,\cdot)$ strictly decreasing and onto $\R$ for each $v \in \R$.
The payoff $f$ is \hyperref[definition:pref_reg]{regular}:
$f_3$ exists, is bounded, and is continuous in $p$,
and I verify the joint continuity property in \cref{suppl:mech:pref_assns_ex}.

A Borel probability measure $y$ on $\Delta(\Omega)$ is the distribution of posteriors induced by some signal exactly if its mean $\int_{\Delta(\Omega)} \mu y( \dd \mu)$ is equal to $\mu_0$.%
	\footnote{The `only if' direction is trivial.
	Conversely, a $y$ with mean $\mu_0$ is induced by a $\Delta(\Omega)$-valued signal whose distribution conditional on each $\omega \in \Omega$ is
	\begin{equation*}
		\pi( M | \omega )
		= \frac{ 1 }{ \mu_0(\omega) } \int_M \mu(\omega) y( \dd \mu )
		\quad \text{for each Borel-measurable $M \subseteq \Delta(\Omega)$.}
	\end{equation*}
	This construction is due to \textcite{Blackwell1951}, and used by \textcite{KamenicaGentzkow2011}.}
Write $\mathcal{Y}$ for the set of all mean-$\mu_0$ distributions of posteriors, and order it by Blackwell informativeness: $y \lesssim y'$ iff
$\smash{\int_{\Delta(\Omega)} v \dd y \leq \int_{\Delta(\Omega)} v \dd y'}$
%
%
for every continuous and convex $v : \Delta(\Omega) \to \R$.%
	\footnote{A Blackwell-less informative distribution of posteriors is precisely one that yields a lower expected payoff
	$\smash{\int_{\raisebox{1pt}{$\scriptstyle \Delta(\Omega)$}} V(\mu,t) y(\dd \mu)}$
	no matter what the underlying action set $A$ or utility $U(\cdot,\cdot,t)$.
	This is because $V(\cdot,t)$ is continuous and convex for any $A$ and $U$,
	and any continuous and convex $v$ can be approximated by $V(\cdot,t)$ for some $A$ and $U$.}
I show in \cref{suppl:mech:order_assns_ex} that the outcome space $\mathcal{Y}$ is \hyperref[definition:outcome_reg]{regular}.

Assume that $f$ satisfies the \hyperref[definition:SM]{strict outer Spence--Mirrlees condition}.
An \emph{information allocation} is a map $Y : [0,1] \to \mathcal{Y}$ that assigns to each type a distribution of posteriors.
By the \hyperref[theorem:incr_impl]{implementability theorem}, we have:

\begin{proposition}
	\label{proposition:info1}
	Every increasing information allocation is implementable.
\end{proposition}

The converse is false.
In particular, there are implementable allocations that assign some types $t < t'$ Blackwell-incomparable information.
But any such information allocation is vulnerable to collusion, as agents of types $t$ and $t'$ would benefit by sharing their information.%
	\footnote{This holds no matter how the underlying signals giving rise to the posterior distributions $Y(t)$ and $Y(t')$ are correlated with each other. For by a standard embedding theorem (e.g. Theorem 7.A.1 in \textcite{ShakedShanthikumar2007}),
	$Y(t) \lesssim Y(t')$ is necessary (as well as sufficient)
	for there to exist a probability space on which there are random vectors with laws $Y(t)$ and $Y(t')$ such that the latter is statistically sufficient for the former.}$^,$%
	\footnote{Both agents benefit \emph{strictly} provided $V(\cdot,t)$ and $V(\cdot,t')$ are strictly convex.}
Call an allocation \emph{sharing-proof} iff no two types are assigned Blackwell-incomparable information.

\begin{proposition}
	\label{proposition:info_sharing}
	An information allocation is implementable and sharing-proof if and only if it is increasing.
\end{proposition}

The proof is in \cref{app:mech:info_sharing_pf}.



\begin{appendices}
\crefalias{section}{appsec}
\crefalias{subsection}{appsec}
\crefalias{subsubsection}{appsec}

\renewcommand*{\thesubsection}{\Alph{subsection}}

\section*{Appendix to the theory (§\ref{sec:background} and §\ref{sec:result})}
\label{app:theory}
\addcontentsline{toc}{section}{Appendix to the theory}

\subsection{Mathematical background}
\label{app:theory:bckg}

Two operations are important in this paper:
writing a function as the integral of its derivative,
and interchanging limits and integrals.
The former is permissible precisely for absolutely continuous functions:
\begin{definition}
	\label{definition:abs_cont}
	A function $\phi : [0,1] \to \R$ is \emph{absolutely continuous} iff for each $\eps > 0$, there is a $\delta > 0$ such that for any finite collection $\{ (r_n,t_n) \}_{n=1}^N$ of disjoint intervals of $[0,1]$,
	$\sum_{n=1}^N ( t_n - r_n ) < \delta$ implies $\sum_{n=1}^N \abs*{ \phi(t_n) - \phi(r_n) } < \eps$.
\end{definition}

Absolute continuity implies continuity and differentiability a.e., but the converse is false. Absolute continuity is implied by Lipschitz continuity.

\begin{namedthm}[Lebesgue's fundamental theorem of calculus.%
	\footnote{See e.g. \textcite[§3.5, p. 106]{Folland1999} for a proof.}]
	\label{theorem:LFTC}
	Let $\phi$ be a function $[0,1] \to \R$. The following are equivalent:
	\begin{enumerate}
		
		\item $\phi$ is absolutely continuous.

		\item There is a $\psi \in \mathcal{L}^1$ such that $\phi(t) = \phi(0) + \int_0^t \psi$ for every $t \in [0,1]$.

		\item $\phi$ is differentiable a.e., its (a.e.-defined) derivative $\phi'$ belongs to $\mathcal{L}^1$, and $\phi(t) = \phi(0) + \int_0^t \phi'$ for every $t \in [0,1]$.
		
	\end{enumerate}
\end{namedthm}

As for interchanging limits and integrals,
uniform integrability is the key:

\begin{definition}
	\label{definition:UI}
	A family $\Phi \subseteq \mathcal{L}^1$ is \emph{uniformly integrable} iff for each $\eps > 0$, there is $\delta > 0$ such that
	for any open $T \subseteq [0,1]$ of measure $<\delta$, we have
	$\int_T \abs*{ \phi } < \eps$ for every $\phi \in \Phi$.
\end{definition}

\begin{namedthm}[Vitali convergence theorem.%
	\footnote{For a proof and a partial converse, see e.g. \textcite[§4.6]{RoydenFitzpatrick2010}.}]
	\label{theorem:Vitali}
	Let $\{ \phi_n \}_{ n \in \N }$ be a uniformly integrable sequence in $\mathcal{L}^1$ converging a.e. to $\phi : [0,1] \to \R$.
	Then $\phi \in \mathcal{L}^1$, and
	$\lim_{ n \to \infty } \int_r^t \phi_n = \int_r^t \phi$ for all $r,t \in [0,1]$.
\end{namedthm}

\noindent
(Lebesgue's dominated convergence theorem is a corollary.)

Absolute continuity and uniform integrability are closely related:
\begin{namedthm}[AC--UI lemma {\normalfont\parencite{FitzpatrickHunt2015}}.]
	\label{lemma:FitzpatrickHunt}
	Let $\phi$ be a continuous function $[0,1] \to \R$. The following are equivalent:
	\begin{enumerate}
		
		\item $\phi$ is absolutely continuous.

		\item The `divided-difference' family $\left\{ t \mapsto [ \phi(t+m) - \phi(t) ] / m \right\}_{m > 0}$
		is uniformly integrable.
		
	\end{enumerate}
\end{namedthm}

\subsection{Housekeeping for absolute equi-continuity (§\ref{sec:background:setting}, p. \pageref{definition:unif_abs_cont})}
\label{app:theory:AEC_HK}

The following lemma justifies the name `absolute equi-continuity',
and is used in \cref{app:theory:necessity_pf} below to prove the \hyperref[lemma:necessity]{necessity lemma} (§\ref{sec:result:oFOC}, p. \pageref{lemma:necessity}).

\begin{lemma}
	\label{lemma:AEC_HK}
	An absolutely equi-continuous family $\{ \phi_x \}_{ x \in \mathcal{X} }$
	is uniformly equi-continuous,
	and each of its members $\phi_x$ is absolutely continuous.
\end{lemma}

\begin{proof}
	Let $\{ \phi_x \}_{x \in \mathcal{X}}$ be absolutely equi-continuous.
	Then for every $x \in \mathcal{X}$,
	$\{ t \mapsto [ \phi_x(t+m) - \phi_x(t) ] / m \}_{ m > 0 }$ is uniformly integrable,
	and hence $\phi_x$ is absolutely continuous by the \hyperref[lemma:FitzpatrickHunt]{AC--UI lemma} in \cref{app:theory:bckg}.
	
	It follows that for any $r < t$ in $[0,1]$,
	\begin{multline*}
		\sup_{x \in \mathcal{X}} \abs*{ \phi_x(t) - \phi_x(r) }
		= \sup_{x \in \mathcal{X}} \abs*{ \int_r^t \phi_x' }
		= \sup_{x \in \mathcal{X}} \abs*{ \lim_{m \downarrow 0} \int_r^t
		\frac{ \phi_x(s+m) - \phi_x(s) }{ m }
		\dd s }
		\\
		\leq \sup_{x \in \mathcal{X}} \sup_{m>0} \abs*{ \int_r^t
		\frac{ \phi_x(s+m) - \phi_x(s) }{ m }
		\dd s }
		\leq \sup_{m>0} \int_r^t \sup_{x \in \mathcal{X}} \abs*{ 
		\frac{ \phi_x(s+m) - \phi_x(s) }{ m } }
		\dd s ,
	\end{multline*}
	where the first equality holds by \hyperref[theorem:LFTC]{Lebesgue's fundamental theorem of calculus},
	and the second holds by the \hyperref[theorem:Vitali]{Vitali convergence theorem}.

	Fix an $\eps>0$.
	By the absolute equi-continuity of $\{ \phi_x \}_{x \in \mathcal{X}}$,
	there is a $\delta>0$ such that whenever $t-r < \delta$,
	the right-hand side of the above inequality is $< \eps$,
	and thus $\sup_{x \in \mathcal{X}} \abs*{ \phi_x(t) - \phi_x(r) } < \eps$.
	So $\{\phi_x\}_{x \in \mathcal{X}}$ is uniformly equi-continuous.
\end{proof}

\subsection{Proof of the \texorpdfstring{\hyperref[lemma:identity_lemma]{identity lemma}}{identity lemma} (§\ref{sec:result:mainthm}, p. \pageref{lemma:identity_lemma})}
\label{app:theory:pf_identity_lemma}

We use the results in \cref{app:theory:bckg}.
We shall focus on the limit $m \downarrow 0$, omitting the symmetric argument for $m \uparrow 0$.%
	\footnote{Since the argument below relies on absolute equi-continuity,
	the omitted argument requires uniform integrability
	of $\{ \Phi_m \}_{m<0} \coloneqq \{ t \mapsto \sup_{x \in \mathcal{X}} \abs*{ ( f(x,t+m) - f(x,t) ) / m } \}_{m<0}$.
	This follows from absolute equi-continuity
	and the observation that $\Phi_m(t) = \Phi_{-m}(t+m)$.}
For $t \in [0,1)$ and $m \in (0,1-t]$, write
\begin{multline*}
	\phi_m(t)
	\coloneqq \frac{ V_X(t+m) - V_X(t) }{ m }
	\\
	= \underbrace{ \frac{ f(X(t+m),t+m) - f(X(t+m),t) }{ m } }
	_{ \textstyle \eqqcolon \psi_m(t) }
	+ \underbrace{ \frac{ f(X(t+m),t) - f(X(t),t) }{ m } }
	_{ \textstyle \eqqcolon \chi_m(t) } .
\end{multline*}

Fix $r,t \in (0,1)$. Note that
\begin{equation*}
	\lim_{ m \downarrow 0 } \int_r^t \chi_m
	= \left. \frac{\dd}{\dd m} \int_r^t f(X(s+m),s) \dd s \right|_{m=0}
\end{equation*}
whenever the limit exists. Our task is to show that $\{ \int_r^t \chi_m \}_{ m > 0 }$ is convergent as $m \downarrow 0$ with limit
\begin{equation*}
	V_X(t) - V_X(r) - \int_r^t f_2(X(s),s) \dd s .
\end{equation*}

$\{ \psi_m \}_{ m > 0 }$ need not converge a.e. under the \hyperref[assumption:basic]{basic assumptions}.%
	\footnote{This remains true even under much stronger assumptions. For example, equi-differentiability of $\{ f(x,\cdot) \}_{ x \in \mathcal{X} }$ is not enough: a counter-example is $\mathcal{X} = [0,1]$, $f(x,t) = (t-x) \1_\Q(x)$ and $X(t)=t$. (Here $\1_\Q(x)=1$ if $x$ is rational and $=0$ otherwise.)
	In this case $\psi_m(t) = \1_\Q(t+m)$, which is nowhere convergent as $m \downarrow 0$.}
But
\begin{equation*}
	\psi^\star_m(t) \coloneqq \frac{ f(X(t),t) - f(X(t),t-m) }{ m } 
\end{equation*}
converges pointwise to $t \mapsto f_2(X(t),t)$, and by a change of variable,
\begin{equation*}
	\int_r^t \psi_m
	= \int_{r+m}^{t+m} \psi^\star_m
	= \int_r^t \psi^\star_m
	+ \left( \int_t^{t+m} \psi^\star_m
	- \int_r^{r+m} \psi^\star_m \right)
	= \int_r^t \psi^\star_m + \oo(1) ,
\end{equation*}
where the bracketed terms vanish as $m \downarrow 0$ because $\{ \psi^\star_m \}_{ m > 0 }$ is uniformly integrable by the \hyperref[assumption:basic]{basic assumptions}.

By absolute continuity of $V_X$ and the \hyperref[lemma:FitzpatrickHunt]{AC--UI lemma} in \cref{app:theory:bckg}, $\{ \phi_m \}_{ m > 0 }$ is uniformly integrable and converges a.e. to $V_X'$ as $m \downarrow 0$.
Since $\{ \psi^\star_m \}_{ m > 0 }$ is uniformly integrable and converges pointwise to $t \mapsto f_2(X(t),t)$, it follows that
\begin{multline*}
	\lim_{ m \downarrow 0 } \int_r^t \chi_m
	= \lim_{ m \downarrow 0 } \int_r^t [ \phi_m - \psi_m ]
	= \lim_{ m \downarrow 0 } \int_r^t [ \phi_m - \psi^\star_m ]
	\\
	= \int_r^t \lim_{ m \downarrow 0 } [ \phi_m - \psi^\star_m ] 
	= \int_r^t \left[ V_X'(s) - f_2(X(s),s) \right] \dd s ,
\end{multline*}
where the third equality holds by the \hyperref[theorem:Vitali]{Vitali convergence theorem}.
Since the last expression is well-defined, this shows $\{ \int_r^t \chi_m \}_{ m > 0 }$ to be convergent as $m \downarrow 0$.
And because $V_X$ is absolutely continuous, the value of the limit is
\begin{equation*}
	\lim_{ m \downarrow 0 } \int_r^t \chi_m
	= V_X(t) - V_X(r) - \int_r^t f_2(X(s),s) \dd s 
\end{equation*}
by \hyperref[theorem:LFTC]{Lebesgue's fundamental theorem of calculus}.
\qed

\subsection{A characterisation of absolute continuity of the value}
\label{app:theory:AC-of-V_charac}

The following lemma characterises the absolute-continuity-of-$V_X$ condition that appears in the \hyperref[theorem:env]{main theorem} (§\ref{sec:result:mainthm}, p. \pageref{theorem:env}). Apart from its independent interest, it is needed
for the proofs in \cref{app:theory:necessity_pf,app:theory:classical_lemma} below.

\begin{lemma}
    \label{lemma:alt_V_abs_cont}
	Under the \hyperref[assumption:basic]{basic assumptions}, the following are equivalent:
	\begin{enumerate}
		
		\item \label{bullet:alt_V_abs_cont:abs}
		$V_X(t) \coloneqq f(X(t),t)$ is absolutely continuous.

		\item \label{bullet:alt_V_abs_cont:UI}
		The family $\{ \chi_m \}_{m>0}$ is uniformly integrable, where
		\begin{equation*}
			\chi_m(t) 
			\coloneqq \frac{ f(X(t+m),t) - f(X(t),t) }{ m } .
		\end{equation*}
	
	\end{enumerate}
\end{lemma}

In the classical case, \ref{bullet:alt_V_abs_cont:UI} is imposed (it follows from the \hyperref[assumption:classical]{classical assumptions}, by \Cref{lemma:classical} in \cref{app:theory:classical_lemma} below). In the modern case, \ref{bullet:alt_V_abs_cont:abs} arises within the theorem. Both are clearly joint restrictions on $f$ and $X$.%
	\footnote{As emphasised by \textcite{MilgromSegal2002}, however,
	any \emph{optimal} $X$ satisfies \ref{bullet:alt_V_abs_cont:abs} provided $f$ satisfies the \hyperref[assumption:basic]{basic assumptions}.
	See \cref{app:theory:necessity_pf} below for a proof.}

\begin{proof}
	Define $\{ \phi_m \}_{m>0}$ and $\{ \psi_m \}_{m>0}$ as in the proof of the \hyperref[lemma:identity_lemma]{identity lemma} (\cref{app:theory:pf_identity_lemma}).
	$\{ \psi_m \}_{ m > 0 }$ is uniformly integrable by the \hyperref[assumption:basic]{basic assumption} of absolute equi-continuity. By the \hyperref[lemma:FitzpatrickHunt]{AC--UI lemma} in \cref{app:theory:bckg}, \ref{bullet:alt_V_abs_cont:abs} is equivalent to $\{ \phi_m \}_{ m > 0 }$ being uniformly integrable.

	Suppose that $\{ \chi_m \}_{ m > 0 }$ is uniformly integrable, and fix $\eps > 0$. Let $\delta > 0$ meet the $\eps/2$-challenge for both $\{ \psi_m \}_{ m > 0 }$ and $\{ \chi_m \}_{ m > 0 }$;
	then for any open $T \subseteq [0,1]$ of measure $<\delta$ and any $m>0$, we have
	\begin{equation*}
		\int_T \abs*{ \phi_m }
		\leq \int_T \abs*{ \psi_m } 
		+ \int_T \abs*{ \chi_m }
		< \frac{\eps}{2} + \frac{\eps}{2}
		= \eps ,
	\end{equation*}
	showing that $\{ \phi_m \}_{ m > 0 }$ is uniformly integrable.

	An almost identical argument establishes that uniform integrability of $\{ \phi_m \}_{ m > 0 }$ implies uniform integrability of $\{ \chi_m \}_{ m > 0 }$.
\end{proof}

\subsection{Proof of the \texorpdfstring{\hyperref[lemma:necessity]{necessity lemma}}{necessity lemma} (§\ref{sec:result:oFOC}, p. \pageref{lemma:necessity})}
\label{app:theory:necessity_pf}

\begin{lemma}
	\label{lemma:AEC_AC}
	If $\{ f(x,\cdot) \}_{ x \in \mathcal{X} }$ is absolutely equi-continuous,
	then the value $V_X(t) \coloneqq f(X(t),t)$
	of any optimal $X : [0,1] \to \mathcal{X}$
	is absolutely continuous.
\end{lemma}

\begin{proof}
	Let $X$ be optimal.
	Then for any $r<t$ in $[0,1)$ and $m \in (0,1-t]$,
	\begin{multline*}
		\abs*{ \frac{1}{m} \int_t^{t+m} V_X
		- \frac{1}{m} \int_r^{r+m} V_X }
		= \abs*{ \int_r^t \frac{ V_X(s+m) - V_X(s) }{ m } \dd s }
		\\
		\leq \int_r^t \abs*{ \frac{ V_X(s+m) - V_X(s) }{ m } } \dd s
		\leq \int_r^t D_m ,
	\end{multline*}
	where
	\begin{equation*}
		D_m(s) \coloneqq \sup_{x \in \mathcal{X}}
		\abs*{ \frac{ f(x,s+m) - f(x,s) }{ m } } .
	\end{equation*}

	Fix an $\eps>0$.
	The absolute equi-continuity of $\{ f(x,\cdot) \}_{x \in \mathcal{X}}$ provides that $\{ D_m \}_{m>0}$ is uniformly integrable,
	so that there is a $\delta>0$ such that for any open $T \subseteq [0,1]$ of measure $<\delta$, we have $\int_T D_m < \eps/2$ for every $m > 0$.
	Thus for any finite collection $\{ (r_n,t_n) \}_{n=1}^N$ of disjoint open intervals of $[0,1]$ whose union $T$ has measure $<\delta$, we have
	\begin{equation*}
		\sum_{n=1}^N \abs*{ \frac{1}{m} \int_{t_n}^{t_n+m} V_X
		- \frac{1}{m} \int_{r_n}^{r_n+m} V_X }
		\leq \int_T D_m
		< \eps/2
		\quad \text{for every $m>0$.}
	\end{equation*}
	$V_X$ is (uniformly) continuous
	since $\{ f(x,\cdot) \}_{x \in \mathcal{X}}$ is uniformly equi-continuous by \Cref{lemma:AEC_HK} in \cref{app:theory:AEC_HK}.%
		\footnote{For any $\eps>0$,
		the uniform equi-continuity of $\{ f(x,\cdot) \}_{x \in \mathcal{X}}$ delivers a $\delta>0$ such that
		$\abs*{t-r} < \delta$ implies
		$\abs*{ V_X(t) - V_X(r) }
		\leq \sup_{x \in \mathcal{X}} \abs*{ f(x,t) - f(x,r) }
		< \eps$.}
	Thus letting $m \downarrow 0$ yields
	\begin{equation*}
		\sum_{n=1}^N \abs*{ V_X(t_n) - V_X(r_n) }
		\leq \eps/2 < \eps 
	\end{equation*}
	by the mean-value theorem, showing $V_X$ to be absolutely continuous.
\end{proof}

\begin{proof}[Proof of the {\hyperref[lemma:necessity]{necessity lemma}}]
	Let $X$ be optimal, and fix $r<t$ in $[0,1]$.
	$V_X$ is absolutely continuous by \Cref{lemma:AEC_AC}.
	Define $\phi_{r,t} : [-r,1-t] \to \R$ by
	\begin{equation*}
		\phi_{r,t}(m) \coloneqq \int_r^t f(X(s+m),s) \dd s 
		\quad \text{for each $m \in [-r,1-t]$.\footnotemark}
			\footnotetext{The map $s \mapsto f(X(s+m),s)$ is integrable because $\abs*{f(X(s+m),s)} \leq \abs*{ V_X(s) } + \abs*{ f(X(s+m),s) - f(X(s),s) }$,
			where the former term is continuous,
			and the latter is integrable by \Cref{lemma:alt_V_abs_cont} in \cref{app:theory:AC-of-V_charac}.}
	\end{equation*}
	$\phi_{r,t}'(0)$ exists by the \hyperref[lemma:identity_lemma]{identity lemma} (§\ref{sec:result:mainthm}, p. \pageref{lemma:identity_lemma}).
	To show that it is zero, observe that for any $s \in (r,t)$ and $m \in (0,\min\{s,1-s\}]$, optimality requires
	\begin{equation*}
		\frac{ f(X(s+m),s) - f(X(s),s) }{ m }
		\leq 0 \leq
		\frac{ f(X(s-m),s) - f(X(s),s) }{ -m } .
	\end{equation*}
	Integrating over $(r,t)$ and letting $m \downarrow 0$ yields $\phi_{r,t}'(0) \leq 0 \leq \phi_{r,t}'(0)$.
\end{proof}

\subsection{A lemma under the \texorpdfstring{\hyperref[assumption:classical]{classical assumptions}}{classical assumptions}}
\label{app:theory:classical_lemma}

The following result is used in the proof of the \hyperref[lemma:housekeeping]{housekeeping lemma} (§\ref{sec:result:oFOC}, p. \pageref{lemma:housekeeping}),
as well as in the proof of the \hyperref[theorem:classical_env]{classical envelope theorem and converse} in \cref{app:theory:classical_env_pf} below.

\begin{lemma}
	\label{lemma:classical}
	Fix a decision rule $X : [0,1] \to \mathcal{X}$, and let
	\begin{equation*}
		\chi_m(t) \coloneqq \frac{ f(X(t+m),t) - f(X(t),t) }{ m } .
	\end{equation*}
	\begin{enumerate}
		
		\item \label{bullet:classical:UI}
		Under the \hyperref[assumption:basic]{basic} and \hyperref[assumption:classical]{classical} assumptions, $\{ \chi_m \}_{ m > 0 }$ is uniformly integrable.

		\item \label{bullet:classical:UI_AC}
		Under the \hyperref[assumption:basic]{basic assumptions}, the following are equivalent:
		\begin{enumerate}
		
			\item $\{ \chi_m \}_{ m > 0 }$ is uniformly integrable and convergent a.e. as $m \downarrow 0$.

			\item $V_X(t) \coloneqq f(X(t),t)$ is absolutely continuous, and
			the derivative $\frac{\dd}{\dd m} f(X(t+m),t) \bigr|_{m=0}$
			exists for a.e. $t \in (0,1)$.
		
		\end{enumerate}
	
	\end{enumerate}
\end{lemma}

\begin{proof}
	For \ref{bullet:classical:UI}, write $K$ for the vector of non-negative constants that bounds $f_1$,
	and $L \geq 0$ for the Lipschitz constant of $X$.
	Let $\norm*{\cdot}_2$ denote the Euclidean norm.
	For any $t \in [0,1)$ and $m \in (0,1-t]$,
	writing $x_\omega \coloneqq (1-\omega) X(t) + \omega X(t+m)$ for $\omega \in [0,1]$,
	we have by the Cauchy--Schwarz inequality that
	\begin{multline*}
		\abs*{ \chi_m(t) }
		= \abs*{ \frac{1}{m} \int_0^1
		\Bigl( f_1\left( x_\omega, t \right)
		\cdot \left[ X(t+m) - X(t) \right] \Bigr) \dd \omega }
		\\
		\leq \frac{1}{m} \int_0^1
		\Bigl( \norm*{ f_1\left( x_\omega, t \right) }_2
		\times \norm*{ X(t+m) - X(t) }_2 \Bigr) \dd \omega
		\leq \frac{1}{m} \norm*{K}_2 \times Lm
		= \norm*{K}_2 L .
	\end{multline*}
	Thus $\{ \chi_m \}_{ m > 0 }$ is uniformly bounded,
	hence uniformly integrable.

	For \ref{bullet:classical:UI_AC}, absolute continuity of $V_X$ is equivalent to uniform integrability of $\{ \chi_m \}_{ m > 0 }$ by \Cref{lemma:alt_V_abs_cont} in \cref{app:theory:AC-of-V_charac}, and a.e. existence of $\frac{\dd}{\dd m} f(X(t+m),t) |_{m=0}$ is definitionally equivalent to a.e. convergence of $\{ \chi_m \}_{ m > 0 }$.
\end{proof}

\subsection{Proof of the \texorpdfstring{\hyperref[theorem:classical_env]{classical envelope theorem and converse}}{classical envelope theorem and converse} (§\ref{sec:background:classical})}
\label{app:theory:classical_env_pf}

\begin{proof}
	Fix a Lipschitz continuous decision rule $X : [0,1] \to \mathcal{X}$.
	By \Cref{lemma:classical} in \cref{app:theory:classical_lemma}, $V_X(t) \coloneqq f(X(t),t)$ is absolutely continuous, hence differentiable a.e.
	The map $r \mapsto f( X(r), t )$ is differentiable a.e. by the \hyperref[assumption:classical]{classical assumptions}, and $t \mapsto f( X(r), t )$ is differentiable by the \hyperref[assumption:basic]{basic assumptions}.
	Hence the a.e.-defined derivative of $V_X$ obeys the differentiation identity
	\begin{equation*}
		V_X'(t) = \left. \frac{\dd}{\dd m} f( X(t+m), t ) \right|_{m=0} + f_2( X(t), t )
		\quad \text{for a.e. $t \in (0,1)$} .
	\end{equation*}
	It follows that the first-order condition a.e. is equivalent to
	\begin{equation*}
		V_X'(t) = f_2( X(t), t )
		\quad \text{for a.e. $t \in (0,1)$} ,
	\end{equation*}
	which in turn is equivalent to the envelope formula by \hyperref[theorem:LFTC]{Lebesgue's fundamental theorem of calculus}.
\end{proof}

By inspection, the proof requires precisely absolute continuity of $V_X$ (so that the envelope formula can be satisfied) and a.e. existence of $\frac{\dd}{\dd m} f(X(t+m),t) |_{m=0}$ (so that the first-order condition a.e. is well-defined).
Part \ref{bullet:classical:UI_AC} of \Cref{lemma:classical} in \cref{app:theory:classical_lemma} therefore tells us that the \hyperref[assumption:classical]{classical assumptions} can be weakened to uniform integrability and a.e. convergence of $\{ \chi_m \}_{ m > 0 }$, and no further.
For $f$ non-trivial, the uniform integrability part involves a strong continuity requirement on $X$.%
	\footnote{For example, consider $\mathcal{X} = [0,1]$, $f(x,t) = x$ and $X(t)=\1_{[r,1]}$, where $r \in (0,1)$.
	Then given $m>0$, we have $\chi_m(t) = 1/m$ for all $t \in [r-m,r]$.
	Suppose toward a contradiction that $\{ \chi_m \}_{ m > 0 }$ is uniformly integrable, and let $\delta>0$ meet the $\eps$-challenge for $\eps \in (0,1)$;
	then for all $m \in ( 0, \delta/2 )$, we have
	%
	$\smash{%
	\int_{\mathrlap{\raisebox{0.5pt}{$\scriptstyle r-\delta/2$}}}%
	{\vphantom{\int}}^{\raisebox{-1pt}{$\scriptstyle r+\delta/2$}}
	\abs*{ \chi_m } 
	\geq 
	\int_{\mathrlap{\raisebox{0.5pt}{$\scriptstyle r-m$}}}%
	{\vphantom{\int}}^{\raisebox{-0.5pt}{$\scriptstyle r\mathrel{\phantom{-n}}$}}
	\abs*{ \chi_m } 
	= m/m
	= 1 > \eps%
	}$,
	which is absurd. This example clearly generalises: the gist is that uniform integrability of $\{ \chi_m \}_{ m > 0 }$ is incompatible with non-removable discontinuities in $X$ unless $f$ is trivial.}

\section*{Appendix to the application (§\ref{sec:app})}
\label{app:mech}
\addcontentsline{toc}{section}{Appendix to the application}

\subsection{Proof of the \texorpdfstring{\hyperref[theorem:incr_impl]{implementability theorem}}{implementability theorem} (§\ref{sec:app:incr_impl}, p. \pageref{theorem:incr_impl})}
\label{app:mech:pf_incr_impl}

We state two lemmata in §\ref{app:mech:pf_incr_impl:solns_env}--§\ref{app:mech:pf_incr_impl:approx}, then prove the theorem in §\ref{app:mech:pf_incr_impl:pf}.

\subsubsection{Solutions of the envelope formula}
\label{app:mech:pf_incr_impl:solns_env}

In the first step of the argument in §\ref{app:mech:pf_incr_impl:pf} below, we are given an allocation $Y$, and wish to choose a payment schedule $P$ such that $(Y,P)$ satisfies the envelope formula.
The following asserts that this can be done:

\begin{namedthm}[Existence lemma.]
	\label{lemma:existence}
	Assume that for all $(y,t) \in \mathcal{Y} \times [0,1]$, $f(y,\cdot,t)$ is strictly decreasing, continuous and onto $\R$.
	Further assume that the type derivative $f_3$ exists and is bounded, and that $f_3(y,\cdot,t)$ is continuous for all $(y,t) \in \mathcal{Y} \times [0,1]$.
	Then for any $k \in \R$ and any allocation $Y : [0,1] \to \mathcal{Y}$ such that $t \mapsto f(Y(t),p,t)$ and $t \mapsto f_3(Y(t),p,t)$ are Borel-measurable for every $p \in \R$, there exists a payment schedule $P : [0,1] \to \R$ such that $(Y,P)$ satisfies the envelope formula with $V_{Y,P}(0)=k$.
\end{namedthm}

\begin{remark}
	\label{remark:existence_elsewhere}
	The following corollary may prove useful elsewhere:
	suppose in addition that $\mathcal{Y}$ is equipped with some topology such that $f(\cdot,p,t)$ and $f_3(\cdot,p,t)$ are Borel-measurable and $f_3(y,p,\cdot)$ is continuous.
	Then for any Borel-measurable allocation $Y : [0,1] \to \mathcal{Y}$, there is a payment schedule $P$ such that $(Y,P)$ satisfies the envelope formula.
\end{remark}

The \hyperref[lemma:existence]{existence lemma} is immediate from the following abstract result by letting $\phi(p,t) \coloneqq f( Y(t), p, t )$ and $\psi(p,t) \coloneqq f_3( Y(t), p, t )$.

\begin{lemma}
	\label{lemma:existence_general}
	Let $\phi$ and $\psi$ be functions $\R \times [0,1] \to \R$. Suppose that $\phi( \cdot, t )$ is strictly decreasing, continuous, and onto $\R$ for every $t \in [0,1]$, and that $\psi$ is bounded with $\psi(\cdot,t)$ continuous for every $t \in [0,1]$.
	Further assume that $\phi(p,\cdot)$ and $\psi(p,\cdot)$ are Borel-measurable for each $p \in \R$.
	Then for any $k \in \R$, there is a function $P : [0,1] \to \R$ such that
	\begin{equation*}
		\phi( P(t), t ) = k + \int_0^t \psi( P(s), s ) \dd s
		\quad \text{for every $t \in [0,1]$} .
	\end{equation*}
\end{lemma}

\begin{proof}
	Since $\phi( \cdot, t )$ is strictly decreasing and continuous, it possesses a continuous inverse $\phi^{-1}(\cdot,t)$, well-defined on all of $\R$ since $\phi(\R,t) = \R$. We may therefore define a function $\chi : \R \times [0,1] \to \R$ by
	\begin{equation*}
		\chi( w, t ) 
		\coloneqq \psi\left( \phi^{-1}( w, t ), t \right) 
		\quad \text{for each $w \in \R$ and $t \in [0,1]$} .
	\end{equation*}
	$\chi(\cdot,t)$ is continuous since $\psi( \cdot, t )$ and $\phi^{-1}(\cdot,t)$ are,
	$\chi$ is bounded since $\psi$ is,
	and $\chi(w,\cdot)$ is Borel-measurable since $\psi(\cdot,t)$ is continuous and $\psi(p,\cdot)$ and $\phi^{-1}(w,\cdot)$ are Borel-measurable.

	Fix $k \in \R$. Consider the integral equation
	\begin{equation*}
		W(t) = k + \int_0^t \chi( W(s), s ) \dd s
		\quad \text{for $t \in [0,1]$} ,
	\end{equation*}
	where $W$ is an unknown function $[0,1] \to \R$. Since $\chi(\cdot,t)$ is continuous and $\chi(w,\cdot)$ bounded and Borel-measurable, there is a local solution by Carathéodory's existence theorem;%
		\footnote{See e.g. Theorem 5.1 in \textcite[ch. 1]{Hale1980}.}
	call it $V$. By boundedness of $\chi$ and a comparison theorem,%
		\footnote{See e.g. Theorem 2.17 in \textcite{Teschl2012}.}
	$V$ can be extended to a solution on all of $[0,1]$.

	Now define $P(t) \coloneqq \phi^{-1}( V(t), t )$. For every $t \in [0,1]$, it satisfies
	\begin{equation*}
		\phi( P(t), t )
		= V(t) 
		= k + \int_0^t \chi( V(s), s ) \dd s 
		= k + \int_0^t \psi( P(s), s ) \dd s . 
		\qedhere 
	\end{equation*}	
\end{proof}

\begin{namedthm}[Uniqueness corollary.]
	\label{corollary:uniqueness}
	Under the hypotheses of the \hyperref[lemma:existence]{existence lemma}, 
	if in addition $\{ f_3(y,\cdot,t) \}_{(y,t) \in \mathcal{Y} \times [0,1]}$ is Lipschitz equi-continuous%
		\footnote{That is, there is an $L \geq 0$ such that $f_3(y,\cdot,t)$ is $L$-Lipschitz for every $(y,t) \in \mathcal{Y} \times [0,1]$.}
	and the monotonicity of $f(y,\cdot,t)$ is uniform in the sense that for some $M>0$,
	\begin{equation*}
		f(y,p,t) - f(y,p',t) \geq M (p'-p)
		\quad \text{for any $p<p'$ in $\R$, $y \in \mathcal{Y}$ and $t \in [0,1]$,}
	\end{equation*}
	then there is \emph{exactly one} payment schedule $P$ such that $(Y,P)$ satisfies the envelope formula with $V_{Y,P}(0)=k$,
	and this payment schedule may be computed via Picard's method.
\end{namedthm}

\begin{proof}
	Again let $\phi(p,t) \coloneqq f( Y(t), p, t )$ and $\psi(p,t) \coloneqq f_3( Y(t), p, t )$,
	and return to the proof of \Cref{lemma:existence_general}.
	The additional assumptions ensure, respectively, that $\{ \psi(\cdot,t) \}_{t \in [0,1]}$ and $\{ \phi^{-1}(\cdot,t) \}_{t \in [0,1]}$ are Lipschitz equi-continuous.
	In follows that $\{ \chi(\cdot,t) \}_{t \in [0,1]}$ is Lipschitz equi-continuous, 
	so that
	(the Picard operator is a contraction, and thus)
	the integral equation has a unique solution
	to which Picard iteration converges in the sup norm.%
		\footnote{See e.g. Theorem 5.3 in \textcite[ch. 1]{Hale1980}.}
\end{proof}

\subsubsection{Continuous approximation of increasing maps}
\label{app:mech:pf_incr_impl:approx}

The second step of the argument §\ref{app:mech:pf_incr_impl:pf} below relies on approximating an increasing map $[0,1] \to \mathcal{Y}$ by continuous and increasing maps.
This is made possible by the following:

\begin{namedthm}[Approximation lemma.]
	\label{lemma:approx}
	Let $\mathcal{Y}$ be \hyperref[definition:outcome_reg]{regular}, and let $Y$ be an increasing map $[0,1] \to \mathcal{Y}$.
	The image $Y([0,1])$ may be embedded in a chain $\mathcal{C} \subseteq \mathcal{Y}$ with $\inf \mathcal{C} = Y(0)$ and $\sup \mathcal{C} = Y(1)$ that is order-dense-in-itself, order-complete and order-separable.%
		\footnote{$\mathcal{C} \subseteq \mathcal{Y}$ is \emph{order-complete} iff every subset with a lower (upper) bound has an infimum (supremum), and \emph{order-separable} iff it has a countable order-dense subset.}
	Furthermore, there exists a sequence $(Y_n)_{n \in \N}$ of increasing maps $[0,1] \to \mathcal{C}$, each with $Y_n = Y$ on $\{0,1\}$,
	such that when $\mathcal{C}$ has the relative topology inherited from the order topology on $\mathcal{Y}$,
	$Y_n$ is continuous for each $n \in \N$,
	and $Y_n \to Y$ pointwise as $n \to \infty$.
\end{namedthm}

The (rather involved) proof is in \cref{suppl:mech:approx_pf}.

\subsubsection{Proof of the \texorpdfstring{\hyperref[theorem:incr_impl]{implementability theorem}}{implementability theorem}}
\label{app:mech:pf_incr_impl:pf}

Fix an increasing $Y : [0,1] \to \mathcal{Y}$.
Embed its image $Y([0,1])$ in the chain $\mathcal{C} \subseteq \mathcal{Y}$ delivered by the \hyperref[lemma:approx]{approximation lemma} in \cref{app:mech:pf_incr_impl:approx}, and equip $\mathcal{C}$ with the relative topology inherited from the order topology on $\mathcal{Y}$.
We henceforth view $Y$ as a function $[0,1] \to \mathcal{C}$, and (with a minor abuse of notation) view $f$ and $f_3$ as functions $\mathcal{C} \times \R \times [0,1] \to \R$.

We seek a payment schedule $P : [0,1] \to \R$ such that the direct mechanism $(Y,P)$ is incentive-compatible.
We do this first (step 1) under the assumption that $Y$ is continuous, then (step 2) show how continuity may be dropped.

\emph{Step 1:}
Suppose that $Y$ is continuous.
By \hyperref[definition:pref_reg]{preference regularity} and the \hyperref[lemma:existence]{existence lemma} in \cref{app:mech:pf_incr_impl:solns_env},%
	\footnote{The measurability hypothesis in the \hyperref[lemma:existence]{existence lemma} is satisfied because $f(\cdot,p,t)$, $f_3(\cdot,p,t)$ and $Y$ are continuous, and $f(y,p,\cdot)$ and $f_3(y,p,\cdot)$ are Borel-measurable (the former being continuous, and the latter a derivative).
	(To complete the argument for measurability, deduce that $r \mapsto f(Y(r),p,t)$ is continuous and that $t \mapsto f(Y(r),p,t)$ is Borel-measurable, so that $(r,t) \mapsto f(Y(r),p,t)$ is (jointly) Borel-measurable, and thus $t \mapsto f(Y(t),p,t)$ is Borel-measurable. Similarly for $f_3$.)}
there exists a payment schedule $P : [0,1] \to \R$ such that the envelope formula holds with (say) $V_{Y,P}(0)=0$:
\begin{equation*}
	V_{Y,P}(t) 
	= \int_0^t f_3( Y(s), P(s), s ) \dd s 
	\quad \text{for every $t \in [0,1]$} .
\end{equation*}
This $P$ must be continuous since $Y$, $f$ and $V_{Y,P}$ are continuous and $f(y,\cdot,t)$ is strictly monotone.%
	\footnote{Suppose not: $t_n \to t$ but $\lim_{ n \to \infty } P(t_n) \neq P(t)$.
	Then the continuity of $Y$ and $f$ and the strict monotonicity of $f(y,\cdot,t)$
	yield a contradiction with the continuity of $V_{Y,P}$:
	\begin{equation*}
		V_{Y,P}(t_n)
		= f( Y(t_n), P(t_n), t_n )
		\to f\left( Y(t), \lim_{n \to \infty} P(t_n), t \right)
		\neq f( Y(t), P(t), t )
		= V_{Y,P}(t) .
	\end{equation*}\vspace{-8pt}}
We will show that $(Y,P)$ is incentive-compatible.

Write $U(r,t) \coloneqq f( Y(r), P(r), t )$ for type $t$'s mimicking payoff,
and $\phi_{r,t}(m) \coloneqq \int_r^t U(s+m,s) \dd s$
for the collective payoff of types $[r,t] \subseteq (0,1)$
from `mimicking up' by $m$.
Clearly $U$ is a continuous function $[0,1]^2 \to \R$,
and thus $\phi_{r,t} : [-r,1-t] \to \R$ is also continuous.
Note that $V_{Y,P}(t) \equiv U(t,t)$.

The model fits into the abstract setting of §\ref{sec:background:setting} by letting $\mathcal{X} \coloneqq \mathcal{C} \times \R$ and $X(t) \coloneqq ( Y(t), P(t) )$, and the \hyperref[assumption:basic]{basic assumptions} are satisfied since $f_3$ exists and is bounded.
We may thus invoke the \hyperref[theorem:env]{converse envelope theorem} (p. \pageref{theorem:env}):
since $(Y,P)$ satisfies the envelope formula,
it must satisfy the outer first-order condition:
\begin{equation*}
	\left. \frac{\dd}{\dd m} \int_{r'}^{t'} U(s+m,s) \dd s \right|_{m=0} = 0
	\quad \text{for all $r'<t'$ in $(0,1)$.}
\end{equation*}

Given $r<t$ in $(0,1)$,
writing $\overline{\DD} \phi_{r,t}(s') \coloneqq \frac{\overline{\dd}}{\overline{\dd} m} \phi_{r,t}(s'+m) \bigr|_{m=0}$ for the upper derivative,
the \hyperref[definition:SM]{outer Spence--Mirrlees condition} yields for each $n \in (0,r)$ that
\begin{align*}
	0
	&\leq
	\left. \frac{\overline{\dd}}{\overline{\dd} m}
	\int_{r-n}^{t-n} U(s+m,s+n) \dd s
	\right|_{m=0}
	\\
	&= \left. \frac{\overline{\dd}}{\overline{\dd} m}
	\int_r^t U(s+m-n,s) \dd s
	\right|_{m=0}
	= \overline{\DD} \phi_{r,t}(-n) ,
\end{align*}
which is to say that $\overline{\DD} \phi_{r,t} \geq 0$ on $(-r,0)$.
Since $\phi_{r,t}$ is continuous, it follows that $\phi_{r,t}$ is increasing on $[-r,0]$.%
	\footnote{This is a standard result; see e.g. \textcite[§11.4, p. 128]{Bruckner1994}.}
A similar argument shows that $\phi_{r,t}$ is decreasing on $[0,1-t]$.

It follows that for any $r<t$ in $[0,1]$ and $m \in [-r,1-t]$,
\begin{equation*}
	\int_r^t [ U(s,s) - U(s+m,s) ] \dd s 
	= \phi_{r,t}(0) - \phi_{r,t}(m)
	\geq 0 .
\end{equation*}
Thus for every $m \in [0,1]$, we have
\begin{equation*}
	U(s,s) - U(s+m,s) \geq 0
	\quad \text{for a.e. $s \in [0,1] \intersect [-m,1-m]$.}
\end{equation*}
Since $s \mapsto U(s,s) = V_{Y,P}(s)$ and
$s \mapsto U(s+m,s)$
are continuous for any $m \in [0,1]$, it follows that for every $m \in [0,1]$,
\begin{equation*}
	U(s,s) - U(s+m,s) \geq 0
	\quad \text{for \emph{every} $s \in [0,1] \intersect [-m,1-m]$,}
\end{equation*}
which is to say that $(Y,P)$ is incentive-compatible.

\emph{Step 2:}
Now drop the assumption that $Y$ is continuous.
By \hyperref[definition:outcome_reg]{regularity of $\mathcal{Y}$} and the \hyperref[lemma:approx]{approximation lemma} in \cref{app:mech:pf_incr_impl:approx}, there exists a sequence $(Y_n)_{n \in \N}$ of continuous and increasing maps $[0,1] \to \mathcal{C}$ converging pointwise to $Y$, each of which satisfies $Y_n=Y$ on $\{0,1\}$.
At each $n \in \N$, Step 1 yields a $P_n : [0,1] \to \R$ such that
such that $(Y_n,P_n)$ is incentive-compatible and satisfies the envelope formula with $V_{Y_n,P_n}(0)=0$.

The sequence $\left( V_{Y_n,P_n} \right)_{n \in \N}$
is Lipschitz equi-continuous%
	\footnote{That is, there is an $L \geq 0$ such that $V_{Y_n,P_n}$ is $L$-Lipschitz for every $n \in \N$.}
by the envelope formula and the boundedness of $f_3$.
It is furthermore uniformly bounded,
due to its Lipschitz equi-continuity and the fact that $V_{Y_n,P_n}(0)=0$ for every $n \in \N$.
Thus by the Arzelà--Ascoli theorem,%
	\footnote{E.g. Theorem 4.44 in \textcite{Folland1999}.}
we may assume (passing to a subsequence if necessary) that $\left( V_{Y_n,P_n} \right)_{n \in \N}$ converges pointwise.
Then $(P_n)_{n \in \N}$ converges pointwise;%
	\footnote{Clearly
	$f( Y_n(t), \inf_{m \geq n} P_m(t), t )
	= \sup_{m \geq n} f( Y_n(t), P_m(t), t )
	\leq \sup_{m \geq n} V_{Y_m,P_m}(t)$
	for any $t \in [0,1]$,
	and thus $f( Y(t), p, t) \leq V(t)$, where $p \coloneqq \liminf_{n \to \infty} P_n(t)$
	and $V(t) \coloneqq \lim_{n \to \infty} V_{Y_n,P_n}(t)$.
	Similarly $V(t) \leq f( Y(t), p', t)$, where $p' \coloneqq \limsup_{n \to \infty} P_n(t)$.
	Thus $f( Y(t), p, t ) \leq f( Y(t), p', t )$,
	which rules out $p<p'$ since $f( Y(t), \cdot, t)$ is strictly decreasing.}
write $P : [0,1] \to \R$ for its limit.

By continuity of $f$,
$U_n(r,t) \coloneqq f( Y_n(r), P_n(r), t )$ converges to $U(r,t) \coloneqq f( Y(r), P(r), t )$ for all $r,t \in [0,1]$.
Each of the incentive-compatibility inequalities $U_n(t,t) \geq U_n(r,t)$ is preserved in the limit $n \to \infty$, ensuring that $(Y,P)$ is incentive-compatible.
\qed

\subsection{Converse to the \texorpdfstring{\hyperref[theorem:incr_impl]{implementability theorem}}{implementability theorem} (§\ref{sec:app:incr_impl}, p. \pageref{theorem:incr_impl})}
\label{app:mech:impl_incr}

In this appendix, we provide a partial converse to the \hyperref[theorem:incr_impl]{implementability theorem},
and use it to prove \Cref{proposition:impl_equiv_incr} (p. \pageref{proposition:impl_equiv_incr}).
We shall use the partial converse again in \cref{app:mech:info_sharing_pf} below to prove \Cref{proposition:info_sharing} (p. \pageref{proposition:info_sharing}).

Letting $\lesssim$ denote the partial order on $\mathcal{Y}$, we say that an allocation $Y : [0,1] \to \mathcal{Y}$ is \emph{non-decreasing} iff there are no $t \leq t'$ in $[0,1]$ such that $Y(t') < Y(t)$.
In other words, $Y(t)$ and $Y(t')$ could either be ranked as $Y(t) \lesssim Y(t')$, or they could be incomparable.
Increasing maps are non-decreasing, but the converse is false except if $\mathcal{Y}$ is a chain.

\begin{namedthm}[\Cref*{proposition:impl_equiv_incr}$\boldsymbol{'}$.]
	\label{proposition:impl_incr}
	If $f$ is \hyperref[definition:pref_reg]{regular} and satisfies the \hyperref[definition:SM]{strict outer Spence--Mirrlees condition},
	then only non-decreasing allocations are implementable.
\end{namedthm}

\begin{proof}[Proof of \Cref{proposition:impl_equiv_incr} (p. \pageref{proposition:impl_equiv_incr})]
	By the \hyperref[theorem:incr_impl]{implementability theorem}, any increasing allocation is implementable.
	By \hyperref[proposition:impl_incr]{\Cref*{proposition:impl_equiv_incr}$'$}, any implementable allocation is non-decreasing, hence increasing since $\mathcal{Y}$ is a chain. 
\end{proof}

The proof of \hyperref[proposition:impl_incr]{\Cref*{proposition:impl_equiv_incr}$'$} relies on two lemmata.
The first is a `non-decreasing' comparative statics result:%
	\footnote{Such results are dimly known in the literature, but rarely seen in print.
	Exceptions include \textcite[Proposition 5]{QuahStrulovici2007extensions} and \textcite{AndersonSmith2021}.}

\begin{lemma}
	\label{lemma:nd_compstat}
	Let $\mathcal{X}$ and $\mathcal{T}$ be partially ordered sets, and let $f$ be a function $\mathcal{X} \times \mathcal{T} \to \R$.
	Call a decision rule $X : \mathcal{T} \to \mathcal{X}$ \emph{optimal} iff $f( X(t), t ) \geq f( x, t )$ for all $x \in \mathcal{X}$ and $t \in \mathcal{T}$.
	If $f$ has strictly single-crossing differences,%
		\footnote{A function $\phi : \mathcal{X} \times \mathcal{T} \to \R$ has \emph{(strictly) single-crossing differences} iff
		$t \mapsto \phi(x',t) - \phi(x,t)$ is (strictly) single-crossing
		for any $x < x'$ in $\mathcal{X}$, where $<$ denotes the strict part of the partial order on $\mathcal{X}$.
		(`Single-crossing' was defined in \cref{footnote:sc_defn} on p. \pageref{footnote:sc_defn}.)}
	then every optimal decision rule is non-decreasing.
\end{lemma}

\begin{proof}
	Write $\lesssim$ and $\preceq$, respectively, for the partial orders on $\mathcal{X}$ and on $\mathcal{T}$.
	Let $X : \mathcal{T} \to \mathcal{X}$ be optimal, and suppose toward a contradiction that there are $t \prec t'$ in $\mathcal{T}$ such that $X(t') < X(t)$.
	Since $X(t)$ is optimal at parameter $t$, we have $f(X(t'),t) \leq f(X(t),t)$.
	Because $t \prec t'$ and $X(t') \prec X(t)$, it follows by strictly single-crossing differences that $f(X(t'),t') < f(X(t),t')$, a contradiction with the optimality of $X(t')$ at parameter $t'$.
\end{proof}

\begin{lemma}
	\label{lemma:outer_ordinal_SM}
	If $f$ is \hyperref[definition:pref_reg]{regular} and satisfies the \hyperref[definition:SM]{(strict) outer Spence--Mirrlees condition},
	then for any price schedule $\pi : \mathcal{Y} \to \R$, the map $(y,t) \mapsto f( y, \pi(y), t )$ has (strictly) single-crossing differences.
\end{lemma}

\begin{proof}
	Fix $y < y'$ in $\mathcal{Y}$, $p,p'$ in $\R$ and $t<t'$ in $[0,1]$.
	Define a mechanism $(Y,P) : [0,1] \to \mathcal{Y} \times \R$ by
	$(Y(s),P(s)) \coloneqq (y,p)$ for $s \leq t$ and
	$(Y(s),P(s)) \coloneqq (y',p')$ for $s > t$,
	and fix $r,r' \in (0,1)$ with $r<t<r'$.
	Clearly for $n \in \{0,t'-t\}$,
	\begin{multline*}
		\left. \frac{\overline{\dd}}{\overline{\dd} m}
		\int_r^{r'} f\left( Y(s+m), P(s+m), s+n \right) \dd s
		\right|_{m=0}
		\\
		\begin{aligned}
			&= \left. \frac{\dd}{\dd m}
			\left(
			\int_r^{t-m}
			f\left( y, p, s+n \right) \dd s
			+ \int_{t-m}^{r'}
			f\left( y', p', s+n \right) \dd s
			\right)
			\right|_{m=0}
			\\
			&= f\left( y', p', t+n \right)
			- f\left( y, p, t+n \right) .			
		\end{aligned}
	\end{multline*}
	If $f$ satisfies the \hyperref[definition:SM]{outer Spence--Mirrlees condition},
	then the left-hand side is single-crossing in $n$,
	and thus $f\left( y', p', t \right) - f\left( y, p, t \right) \geq \mathrel{(>)} 0$ implies $f\left( y', p', t' \right) - f\left( y, p, t' \right) \geq \mathrel{(>)} 0$.
	Similarly for the strict case.
\end{proof}

\begin{proof}[Proof of {\hyperref[proposition:impl_incr]{\Cref*{proposition:impl_equiv_incr}$\xslantmath{'}$}}]
	Let $Y : [0,1] \to \mathcal{Y}$ be implementable, so that $(Y,P)$ is incentive-compatible for some payment schedule $P : [0,1] \to \R$.
	Define a price schedule $\pi : Y([0,1]) \to \R$ by $\pi \circ Y = P$; it is well-defined because by incentive-compatibility and strict monotonicity of $f(y,\cdot,t)$, $Y(r)=Y(r')$ implies $P(r)=P(r')$.
	Define a function $\phi : Y([0,1]) \times [0,1] \to \R$ by $\phi(y,t) \coloneqq f(y,\pi(y),t)$.
	Take any $t \in [0,1]$ and $y \in Y([0,1])$, and observe that there must be an $r \in [0,1]$ with $Y(r) = y$. Then since $(Y,P)$ is incentive-compatible,
	\begin{multline*}
		\phi( Y(t), t )
		= f( Y(t), \pi(Y(t)), t )
		= f( Y(t), P(t), t )
		\\
		\geq f( Y(r), P(r), t )
		= f( y, \pi(y), t ) 
		= \phi( y, t ) .
	\end{multline*}
	Since $y \in Y([0,1])$ and $t \in [0,1]$ were arbitrary, this shows that $Y$ is an optimal decision rule for objective $\phi$.
	Since $\phi$ has strictly single-crossing differences by \Cref{lemma:outer_ordinal_SM}, it follows by \Cref{lemma:nd_compstat} that $Y$ is non-decreasing.
\end{proof}

\subsection{Proof of Proposition \ref{proposition:info_sharing} (§\ref{sec:app:info}, p. \pageref{proposition:info_sharing})}
\label{app:mech:info_sharing_pf}

Any increasing $Y : [0,1] \to \mathcal{Y}$ is implementable by the \hyperref[theorem:incr_impl]{implementability theorem} (§\ref{sec:app:incr_impl}, p. \pageref{theorem:incr_impl}), and clearly sharing-proof.
For the converse, let $Y : [0,1] \to \mathcal{Y}$ be implementable and sharing-proof, and fix $t < t'$; 
then either $Y(t) \lesssim Y(t')$ or $Y(t') < Y(t)$ since $Y$ is sharing-proof,
and it cannot be the latter because $Y$ is non-decreasing by \hyperref[proposition:impl_incr]{\Cref*{proposition:impl_equiv_incr}$'$} in \cref{app:mech:impl_incr}.
\qed

\crefalias{section}{supplsec}
\crefalias{subsection}{supplsec}
\crefalias{subsubsection}{supplsec}

\section*{Supplemental appendix to the application (§\ref{sec:app})}
\label{suppl:mech}
\addcontentsline{toc}{section}{Supplemental appendix to the application}

\subsection{The failure of the standard implementability argument}
\label{suppl:mech:standar_arg}

When the agent's preferences have the quasi-linear form $f(y,p,t) = h(y,t) - p$,
a standard argument establishes the implementability of increasing allocations
without resort to the converse envelope theorem.
I first outline the argument, then show how it fails absent quasi-linearity,
necessitating my alternative approach based on the converse envelope theorem.

Fix an increasing allocation $Y : [0,1] \to \mathcal{Y}$.
Choose a $P$ so that $(Y,P)$ satisfies the envelope formula.%
	\footnote{In the quasi-linear case, such a $P$ is given explicitly by
	$P(t) \coloneqq h( Y(t), t ) - \int_0^t h_2(Y(s),s) \dd s$,
	obviating the need to invoke the \hyperref[lemma:existence]{existence lemma} in \cref{app:mech:pf_incr_impl:solns_env}.}
We then have for any $r,t \in [0,1]$ that
\begin{multline*}
	f(Y(t),P(t),t) - f(Y(r),P(r),t)
	\\
	\begin{aligned}
		&= [ V_{Y,P}(t) - V_{Y,P}(r) ]
		- [ f(Y(r),P(r),t) - f(Y(r),P(r),r) ]
		\\
		&= \int_r^t \left[ f_3(Y(s),P(s),s) - f_3(Y(r),P(r),s) \right] \dd s 
	\end{aligned}
\end{multline*}
by the envelope formula
and \hyperref[theorem:LFTC]{Lebesgue's fundamental theorem of calculus}.

For quasi-linear preferences,
$f_3(y,p,s)$ does not vary with $p$,
and $f$ is single-crossing iff $y \mapsto f_3(y,0,s)$ is increasing for every $s \in [0,1]$.%
	\footnote{This is easily shown,
	and does not depend on exactly how `single-crossing' is formalised.}
Since $Y$ is also increasing,
this implies that the above integrand is non-negative,
which (since $r,t \in [0,1]$ were arbitrary) shows that $(Y,P)$ is incentive-compatible.

These properties of quasi-linearity are very special, however.
In general, single-crossing has nothing directly to say about the type derivative $f_3$, and so cannot be used to sign the integrand.
The standard argument thus fails.

The argument may of course be salvaged
by replacing single-crossing
with the brute assumption that the integrand is non-negative.
But this assumption lacks a choice interpretation,
being a restriction on the \emph{type} derivative $f_3$ of the utility representation $f$.
A theorem with such a hypothesis would have no economic meaning.
(By contrast, single-crossing has a straightforward choice interpretation, described in the text.)

\subsection{Some \texorpdfstring{\hyperref[definition:outcome_reg]{regular}}{regular} outcome spaces (§\ref{sec:app:reg_sc})}
\label{suppl:mech:order_assns_ex}

\begin{proposition}
	\label{proposition:order_assns_examples}
	The following partially ordered sets are \hyperref[definition:outcome_reg]{regular}:
	\begin{enumerate}[label=(\alph*)]
	
		\item \label{item:order_assns_examples:Rn}
		$\R^n$ equipped with the usual (product) order: $(y_1,\dots,y_n) \lesssim (y_1',\dots,y_n')$ iff $y_i \leq y_i'$ for every $i \in \{1,\dots,n\}$.

		\item \label{item:order_assns_examples:l1}
		The space $\ell^1$ of summable sequences equipped with the product order: $(y_i)_{i \in \N} \lesssim (y_i')_{i \in \N}$ iff $y_i \leq y_i'$ for every $i \in \N$.

		\item \label{item:order_assns_examples:L1}
		For any measure space $(\Omega,\mathcal{F},\mu)$,
		the space $\mathcal{L}^1(\Omega,\mathcal{F},\mu)$ of (equivalence classes of $\mu$-a.e. equal) $\mu$-integrable functions $\Omega \to \R$, equipped with the partial order $\lesssim$ defined by
		$y \lesssim y'$ iff $y \leq y'$ $\mu$-a.e.

		(Special case: for any probability space, the space of finite-expectation random variables, ordered by `a.s. smaller'.)

		\item \label{item:order_assns_examples:distbel}
		For any finite set $\Omega$ and probability $\mu_0 \in \Delta(\Omega)$, the space of mean-$\mu_0$ Borel probability measures on $\Delta(\Omega)$,
		equipped with the Blackwell informativeness order defined in §\ref{sec:app:info}.%
			\footnote{A proof that this is a partial order (in particular, anti-symmetric) may be found in \textcite[Theorem 5.2]{Muller1997}.}

		\item \label{item:order_assns_examples:intervals}
		The open intervals of $(0,1)$ (including $\varnothing$), ordered by set inclusion $\subseteq$.
	
	\end{enumerate}
\end{proposition}

We will use the following sufficient condition for chain-separability.

\begin{lemma}
	\label{lemma:order_embedding}
	If there is a strictly increasing function $\mathcal{Y} \to \R$, then $\mathcal{Y}$ is chain-separable.
\end{lemma}

(The converse is false: there are chain-separable spaces that admit no strictly increasing real-valued function.)

\begin{proof}
	Suppose that $\phi : \mathcal{Y} \to \R$ is a strictly increasing function, and let $Y \subseteq \mathcal{Y}$ be a chain; we will show that $Y$ has a countable order-dense subset.
	By inspection, the restriction $\phi|_Y$ of $\phi$ to $Y$ is an order-embedding of $Y$ into $\R$; thus $Y$ is order-isomorphic to a subset of $\R$ (namely $\phi(Y)$).
	The order-isomorphs of subsets of $\R$ are precisely those chains that have a countable order-dense subsets (see e.g. Theorem 24 in \textcite[p. 200]{Birkhoff1967}); thus $Y$ has a countable order-dense subset.
\end{proof}

\begin{proof}[Proof of \Cref{proposition:order_assns_examples}\ref{item:order_assns_examples:Rn}--\ref{item:order_assns_examples:L1}]
	$\R^n$ is exactly $\mathcal{L}^1\bigl( \{1,\dots,n\}, 2^{\{1,\dots,n\}}, c \bigr)$
	where $c$ is the counting measure; similarly, $\ell^1$ is $\mathcal{L}^1\left( \N, 2^\N, c \right)$.
	It therefore suffices to establish \ref{item:order_assns_examples:L1}.

	So fix a measure space $(\Omega,\mathcal{F},\mu)$, and let $\mathcal{Y} \coloneqq \mathcal{L}^1(\Omega,\mathcal{F},\mu)$ be ordered by `$\mu$-a.e. smaller'.
	$\mathcal{Y}$ is order-dense-in-itself since if $y \leq y''$ $\mu$-a.e. and $y \neq y''$ on a set of positive $\mu$-measure, then $y' \coloneqq (y+y'')/2$ lives in $\mathcal{Y}$ and satisfies $y \leq y' \leq y''$ $\mu$-a.e. and $y \neq y' \neq y''$ on a set of positive $\mu$-measure.

	For countable-chain completeness, take any countable chain $Y \subseteq \mathcal{Y}$, and suppose that it has a lower bound $y \in \mathcal{Y}$; we will show that $Y$ has an infimum. (The argument for upper bounds is symmetric.)
	Define $y_\star : \Omega \to \R$ by $y_\star(\omega) \coloneqq \inf_{y \in Y} y(\omega)$ for each $\omega \in \Omega$; it is well-defined (i.e. it maps into $\R$, with the possible exception of a $\mu$-null set) since $Y$ has a lower bound.
	Clearly $y' \leq y_\star \leq y''$ $\mu$-a.e. for any lower bound $y'$ of $Y$ and any $y'' \in Y$, so it remains only to show that $y_\star$ lives in $\mathcal{Y}$, meaning that it is measurable and that its integral is finite.
	Measurability obtains since $Y$ is countable (e.g. Proposition 2.7 in \textcite{Folland1999}).
	As for the integral, since $y \leq y_\star \leq y_0$ $\mu$-a.e. and $y$ and $y_0$ are integrable (live in $\mathcal{Y}$), we have
	\begin{equation*}
		-\infty 
		< \int_\Omega y \dd \mu
		\leq \int_\Omega y_\star \dd \mu
		\leq \int_\Omega y_0 \dd \mu
		< +\infty .
	\end{equation*}

	For chain-separability, define $\phi : \mathcal{Y} \to \R$ by
	$\phi(y) \coloneqq \int_\Omega y \dd \mu$ for each $y \in \mathcal{Y}$.
	$\phi$ is strictly increasing: if $y \leq y'$ $\mu$-a.e. and $y \neq y'$ on a set of positive $\mu$-measure, then $\phi(y) < \phi(y')$.
	Chain-separability follows by \Cref{lemma:order_embedding}.
\end{proof}

\begin{proof}[Proof of \Cref{proposition:order_assns_examples}\ref{item:order_assns_examples:distbel}]
	Fix a finite set $\Omega$ and a probability $\mu_0 \in \Delta(\Omega)$, and let $\mathcal{Y}$ be the space of Borel probability measures with mean $\mu_0$, equipped with the Blackwell informativeness order $\lesssim$.
	$\mathcal{Y}$ is order-dense-in-itself because if $y,y'' \in \mathcal{Y}$ satisfy $\int_{\Delta(\Omega)} v \dd y \leq \int_{\Delta(\Omega)} v \dd y''$ for every continuous and convex $v : \Delta(\Omega) \to \R$, with the inequality strict for some $v = \widehat{v}$,
	then $y' \coloneqq (y+y'')/2$ also lives in $\mathcal{Y}$ and satisfies
	$\int_{\Delta(\Omega)} v \dd y 
	\leq \int_{\Delta(\Omega)} v \dd y'
	\leq \int_{\Delta(\Omega)} v \dd y''$
	for every continuous and convex $v : \Delta(\Omega) \to \R$, with both inequalities strict for $v = \widehat{v}$.

	For countable chain-completeness, let $Y \subseteq \mathcal{Y}$ be a countable chain with an upper bound in $\mathcal{Y}$; we will show that it has a supremum. (The argument for infima is analogous.)
	This is trivial if $Y$ has a maximum element, so suppose not.
	Then there is a strictly increasing sequence $(y_n)_{n \in \N}$ in $Y$ that has no upper bound in $Y$.
	This sequence is trivially tight since $\Delta(\Omega)$ is a compact metric space, so has a weakly convergent subsequence $(y_{n_k})_{k \in \N}$ by Prokhorov's theorem;%
		\footnote{E.g. Theorem 5.1 in \textcite{Billingsley1999}.}
	call the limit $y^\star$.
	Then by the monotone convergence theorem for real numbers and the definition of weak convergence, we have for every for every continuous (hence bounded) and convex $v : \Delta(\Omega) \to \R$ that
	\begin{equation*}
		\sup_{y \in Y} \int_{\Delta(\Omega)} v \dd y
		= \lim_{k \to \infty} \int_{\Delta(\Omega)} v \dd y_{n_k}
		= \int_{\Delta(\Omega)} v \dd y^\star ,
	\end{equation*}
	which is to say that $y^\star$ is the supremum of $Y$.

	For chain-separability, it suffices by \Cref{lemma:order_embedding} to identify a strictly increasing function $\mathcal{Y} \to \R$.
	Let $v$ be any strictly convex function $\Delta(\Omega) \to \R$,%
		\footnote{E.g. the $\mathcal{L}^2$ norm $\norm*{\cdot}_2$, which is strictly convex on $\Delta(\Omega)$ by Minkowski's inequality.}
	and define $\phi : \mathcal{Y} \to \R$ by $\phi(y) \coloneqq \int_{\Delta(\Omega)} v \dd y$.
	Take $y < y'$ in $\mathcal{Y}$; we must show that $\phi(y) < \phi(y')$.
	By a standard embedding theorem (e.g. Theorem 7.A.1 in \textcite{ShakedShanthikumar2007}), there exists a probability space on which there are random vectors $X,X'$ with respective laws $y,y'$ such that $\E(X'|X) = X$ a.s. and $X \neq X'$ with positive probability.
	Thus
	\begin{equation*}
		\phi(y')
		= \E( v(X') )
		= \E( \E[ v(X') | X ] )
		> \E( v( \E[ X' | X ] ) )
		= \E( v(X) )
		= \phi(y)
	\end{equation*}
	by Jensen's inequality.
\end{proof}

\begin{proof}[Proof of \Cref{proposition:order_assns_examples}\ref{item:order_assns_examples:intervals}]
	%
	Write $\mathcal{Y}$ for the open intervals of $(0,1)$.
	$\mathcal{Y}$ is order-dense-in-itself since if $(a,b) \subsetneq (a'',b'')$ then
	$(a',b') \coloneqq \left( [a+a''] / 2, [b+b''] / 2 \right)$
	is an open interval (lives in $\mathcal{Y}$) and satisfies $(a,b) \subsetneq (a',b') \subsetneq (a'',b'')$.

	For countable chain-completeness, we must show that every countable chain has an infimum and supremum.
	So take a countable chain $Y \subseteq \mathcal{Y}$, define $y^\star \coloneqq \Union_{y \in Y} y$, and let $y_\star$ be the interior of $\Intersect_{y \in Y} y$.
	Both are open intervals, so live in $\mathcal{Y}$.
	Clearly $y \subseteq y^\star \subseteq y^+$ for any $y \in Y$ and any set $y^+$ containing every member of $Y$, so $y^\star$ is the supremum of $Y$.
	Similarly $y_\star \subseteq \Intersect_{y' \in Y} y' \subseteq y$ for any $y \in Y$, and $y_- \subseteq y_\star$ for any \emph{open} set $y_-$ contained in every member of $Y$ since $y_\star$ is by definition the $\subseteq$-largest open set contained in $\Intersect_{y \in Y} y$.

	For chain-separability, define $\phi : \mathcal{Y} \to \R$ by
	$\phi((a,b)) \coloneqq b-a$.
	It is clearly strictly increasing, giving us chain-separability by \Cref{lemma:order_embedding}.
\end{proof}

\subsection{Proof of the \texorpdfstring{\hyperref[lemma:approx]{approximation lemma}}{approximation lemma} (appendix \ref{app:mech:pf_incr_impl:approx})}
\label{suppl:mech:approx_pf}

Let $Y : [0,1] \to \mathcal{Y}$ be increasing.
Then $Y([0,1])$ is a chain.
The result is trivial if $Y([0,1])$ is a singleton, so suppose not.

We will first show (steps 1--3) that $Y([0,1])$ may be embedded in a chain $\mathcal{C} \subseteq \mathcal{Y}$ with $\inf \mathcal{C} = Y(0)$ and $\sup \mathcal{C} = Y(1)$ that is order-dense-in-itself, order-complete and order-separable.
We will then argue (step 4) that this chain $\mathcal{C}$ is order-isomorphic and homeomorphic to the unit interval, allowing us to treat $Y$ as a function $[0,1] \to [0,1]$.

\emph{Step 1: construction of $\mathcal{C}$.}
Write $\lesssim$ for the partial order on $\mathcal{Y}$.
Define $\mathcal{Y}'$ to be the set of all outcomes $y' \in \mathcal{Y}$ that are $\lesssim$-comparable to every $y \in Y([0,1])$ and that satisfy $Y(0) \lesssim y' \lesssim Y(1)$.

We claim that $\mathcal{Y}'$ is order-dense-in-itself.
Suppose to the contrary that there are $y < y''$ in $\mathcal{Y}'$ for which no $y' \in \mathcal{Y}'$ satisfies $y < y' < y''$.
Observe that by definition of $\mathcal{Y}'$, any $x \in Y([0,1])$ must be comparable to both $y$ and $y''$, so that
\begin{equation*}
	\left\{ x \in Y([0,1]) : 
	\text{$x \lesssim y$ or $y'' \lesssim x$} \right\}
	= Y([0,1]) .
\end{equation*}
Since it is order-dense-in-itself, the grand space $\mathcal{Y}$ does contain an outcome $y'$ such that $y < y' < y''$.
Since $\lesssim$ is transitive (being a partial order), it follows that $y'$ is comparable to every element of
\begin{equation*}
	\left\{ x \in \mathcal{Y} : 
	\text{$x \lesssim y$ or $y'' \lesssim x$} \right\}
	\supseteq 
	\left\{ x \in Y([0,1]) : 
	\text{$x \lesssim y$ or $y'' \lesssim x$} \right\}
	= Y([0,1]) .
\end{equation*}
But then $y'$ lies in $\mathcal{Y}'$ by definition of the latter---a contradiction.

Clearly $Y(1)$ is an upper bound of any chain in $\mathcal{Y}'$.
It follows by the Hausdorff maximality principle (which is equivalent to the Axiom of Choice) that there is a chain $\mathcal{C} \subseteq \mathcal{Y}'$ that is maximal with respect to set inclusion. (That is, $\mathcal{C} \union \{y\}$ fails to be a chain for every $y \in \mathcal{Y}' \setminus \mathcal{C}$.)

\emph{Step 2: easy properties of $\mathcal{C}$.}
By definition of $\mathcal{Y}'$, any maximal chain in $\mathcal{Y}'$ (in particular, $\mathcal{C}$) contains $Y([0,1])$ and has infimum $Y(0)$ and supremum $Y(1)$.

To see that $\mathcal{C}$ is order-dense-in-itself, assume toward a contradiction that there are $c < c''$ for which no $c' \in \mathcal{C}$ satisfies $c < c' < c''$, so that (since $\mathcal{C}$ is a chain)
\begin{equation*}
	\{ c' \in \mathcal{C} : c' \lesssim c \}
	\union
	\{ c' \in \mathcal{C} : c'' \lesssim c' \}
	= \mathcal{C} .
\end{equation*}
Because $\mathcal{Y}'$ is order-dense-in-itself, there is a $y' \in \mathcal{Y}' \setminus \mathcal{C}$ with $c < y' < c''$.
It follows by transitivity of $\lesssim$ that $y'$ is comparable to every element of
\begin{equation*}
	\{ c' \in \mathcal{C} : c' \lesssim c \}
	\union
	\{ c' \in \mathcal{C} : c'' \lesssim c' \}
	= \mathcal{C} .
\end{equation*}
But then $\mathcal{C} \union \{y'\}$ is a chain in $\mathcal{Y}'$, contradicting the maximality of $\mathcal{C}$.

To establish that $\mathcal{C}$ is order-separable, we must find a countable order-dense subset of $\mathcal{C}$.
Because the grand space $\mathcal{Y}$ is chain-separable, it contains a countable set $\mathcal{K}$ that is order-dense in $\mathcal{C}$.
Since $\mathcal{C}$ is a chain contained in
\begin{equation*}
	\left\{ y \in \mathcal{Y} : Y(0) \lesssim y \lesssim Y(1) \right\} ,
\end{equation*}
we may assume without loss of generality that every $k \in \mathcal{K}$ satisfies $Y(0) \lesssim k \lesssim Y(1)$ and is comparable to every element of $\mathcal{C}$.
It follows that $\mathcal{K}$ is contained in $\mathcal{Y}'$ (by definition of the latter).
We claim that $\mathcal{K}$ is contained in $\mathcal{C}$.
Suppose to the contrary that there is a $k \in \mathcal{K}$ that does not lie in $\mathcal{C}$; then $\mathcal{C} \union \{k\}$ is a chain in $\mathcal{Y}'$, which is absurd since $\mathcal{C}$ is maximal.

\emph{Step 3: order-completeness of $\mathcal{C}$.}
Since every subset of $\mathcal{C}$ has a lower and an upper bound (viz. $Y(0)$ and $Y(1)$, respectively), what must be shown is that every subset of the chain $\mathcal{C}$ has an infimum and a supremum in $\mathcal{C}$.
To that end, take any subset $\mathcal{C}'$ of $\mathcal{C}$, necessarily a chain.

We will first (step 3(a)) show that if $\inf \mathcal{C}'$ exists in $\mathcal{Y}$, then it must lie in $\mathcal{C}$.
We will then (step 3(b)) construct a countable chain $\mathcal{C}''' \subseteq \mathcal{C}'$, for which $\inf \mathcal{C}'''$ exists in $\mathcal{Y}$ by countable-chain completeness of $\mathcal{Y}$, and show that it is also the infimum in $\mathcal{Y}$ of $\mathcal{C}'$.
We omit the analogous arguments for $\sup \mathcal{C}'$.

\emph{Step 3(a): $\inf \mathcal{C}' \in \mathcal{C}$ if the former exists in $\mathcal{Y}$.}
Suppose that $\inf \mathcal{C}'$ exists in $\mathcal{Y}$.
We claim that it lies in $\mathcal{Y}'$, meaning
that $Y(0) \lesssim \inf \mathcal{C}' \lesssim Y(1)$ and that $\inf \mathcal{C}'$ is comparable to every $y \in Y([0,1])$.
The former condition is clearly satisfied.
For the latter, since $\inf \mathcal{C}'$ is a lower bound of $\mathcal{C}'$, transitivity of $\lesssim$ ensures that it is comparable to every $y \in Y([0,1])$ such that $c' \lesssim y$ for some $c' \in \mathcal{C}'$.
To see that $\inf \mathcal{C}'$ is also comparable to every $y \in Y([0,1])$ with $y < c'$ for every $c' \in \mathcal{C}'$, note that any such $y$ is a lower bound of $\mathcal{C}'$.
Since $\inf \mathcal{C}'$ is the \emph{greatest} lower bound, we must have $y \lesssim \inf \mathcal{C}'$, showing that $\inf \mathcal{C}'$ is comparable to $y$.

Now to show that $\inf \mathcal{C}'$ lies in $\mathcal{C}$, decompose the chain $\mathcal{C}$ as
\begin{align*}
	\mathcal{C}
	&= \{ c \in \mathcal{C} : \text{$c \lesssim c'$ for every $c' \in \mathcal{C}'$} \}
	\union \{ c \in \mathcal{C} : \text{$c' < c$ for some $c' \in \mathcal{C}'$} \}
	\\
	&= \{ c \in \mathcal{C} : c \lesssim \inf \mathcal{C}' \}
	\union \{ c \in \mathcal{C} : \inf \mathcal{C}' < c \} .
\end{align*}
Clearly $\inf \mathcal{C}'$ is comparable to every element of $\mathcal{C}$, and we showed that it lies in $\mathcal{Y}'$.
Thus $\mathcal{C} \union \{\inf \mathcal{C}'\}$ is a chain in $\mathcal{Y}'$, which by maximality of $\mathcal{C}$ requires that $\inf \mathcal{C}' \in \mathcal{C}$.

\emph{Step 3(b): $\inf \mathcal{C}'$ exists in $\mathcal{Y}$.}
By essentially the same construction as we used to embed $Y([0,1])$ in $\mathcal{Y}'$ in step 1, $\mathcal{C}'$ may be embedded in a chain $\mathcal{C}'' \subseteq \mathcal{C}$ that is order-dense-in-itself such that for every $c'' \in \mathcal{C}''$, we have $c'_- \lesssim c'' \lesssim c'_+$ for some $c'_-,c'_+ \in \mathcal{C}'$.
By order-separability of $\mathcal{C}$, $\mathcal{C}''$ has a countable order-dense subset $\mathcal{C}'''$, necessarily a chain.
By countable chain-completeness of $\mathcal{Y}$, $\inf \mathcal{C}'''$ exists in $\mathcal{Y}$.
We will show that it is the greatest lower bound of $\mathcal{C}'$.

Observe that $\inf \mathcal{C}'''$ is a lower bound of $\mathcal{C}''$ since $\mathcal{C}'''$ is order-dense in $\mathcal{C}''$.
There can be no greater lower bound of $\mathcal{C}''$ since $\mathcal{C}''' \subseteq \mathcal{C}''$.
Thus $\inf \mathcal{C}''$ exists in $\mathcal{Y}$ and equals $\inf \mathcal{C}'''$.

Since $\inf \mathcal{C}''$ is a lower bound of $\mathcal{C}'' \supseteq \mathcal{C}'$, it is a lower bound of $\mathcal{C}'$.
On the other hand, by construction of $\mathcal{C}''$, we may find for every $c'' \in \mathcal{C}''$ a $c' \in \mathcal{C}'$ such that $c' \lesssim c''$, so there cannot be a greater lower bound of $\mathcal{C}'$.
Thus $\inf \mathcal{C}''$ is the greatest lower bound of $\mathcal{C}'$ in $\mathcal{Y}$.

\emph{Step 4: identification of $\mathcal{C}$ with $[0,1]$.}
Since $\mathcal{C}$ is an order-separable chain, it is order-isomorphic to a subset $\mathcal{S}$ of $\R$ (see e.g. Theorem 24 in \textcite[p. 200]{Birkhoff1967}).
It follows that $\mathcal{C}$ with the order topology is homeomorphic to $\mathcal{S}$ with its order topology.

The set $\mathcal{S}$ is dense in an interval $\mathcal{S}' \supseteq \mathcal{S}$ since $\mathcal{S}$ is order-dense-in-itself (because $\mathcal{C}$ is).
The interval $\mathcal{S}'$ must be closed and bounded since it contains its infimum and supremum (because $\mathcal{C}$ contains $Y(0)$ and $Y(1)$).
Since $\mathcal{S}$ is order-complete (because $\mathcal{C}$ is), it must coincide with its closure, so that $\mathcal{S}'=\mathcal{S}$.
Finally, $\mathcal{S}$ is a proper interval since $\mathcal{C}$ is neither empty nor a singleton.
In sum, we may identify $\mathcal{C}$ with a closed and bounded proper interval of $\R$---without loss of generality, the unit interval $[0,1]$.


We may therefore treat $Y$ as an increasing function $[0,1] \to [0,1]$.
With this simplification, it is straightforward to construct a sequence $(Y_n)_{n \in \N}$ with the desired properties; we omit the details.
\qed

\subsection{\texorpdfstring{\hyperref[definition:pref_reg]{Preference regularity}}{Preference regularity} in selling information (§\ref{sec:app:info})}
\label{suppl:mech:pref_assns_ex}

In this appendix, we show that the joint continuity part of \hyperref[definition:pref_reg]{preference regularity} (p. \pageref{definition:pref_reg}) is satisfied in §\ref{sec:app:info}.
We require two lemmata.

\begin{lemma}
	\label{lemma:pref_assns_ex_conv}
	Let $\mathcal{Y}$ be the set of Borel probability distributions with mean $\mu_0$, equipped with the Blackwell informativeness order (as in §\ref{sec:app:info}).
	Give $\mathcal{Y}$ the order topology, and let $\mathcal{C} \subseteq \mathcal{Y}$ be a chain.
	If a sequence $(y_n)_{n \in \N}$ in $\mathcal{C}$ converges to $y \in \mathcal{C}$ in the relative topology on $\mathcal{C}$, then
	\begin{equation*}
		\sup_{ \substack{ v^+,v^- : \Delta(\Omega) \to \R
		\\ \text{continuous convex} \\ \text{s.t. $\abs{v^+-v^-} \leq 1$} } } 
		\abs*{ \int_{\Delta(\Omega)} \bigl( v^+ - v^- \bigr) \dd (y_n-y) }
		\to 0
		\quad \text{as $n \to \infty$.}
	\end{equation*}
\end{lemma}

\begin{corollary}
	\label{corollary:pref_assns_ex_conv}
	Under the same hypotheses,
	\begin{equation*}
		\sup_{ \substack{ v : \Delta(\Omega) \to [-1,1]
		\\ \text{continuous convex} } }
		\abs*{ \int_{\Delta(\Omega)} v \dd (y_n-y) }
		\to 0
		\quad \text{as $n \to \infty$.}
	\end{equation*}
\end{corollary}

\begin{proof}[Proof of \Cref{lemma:pref_assns_ex_conv}]
	Define $d : \mathcal{Y} \times \mathcal{Y} \to \R_+$ by
	\begin{equation*}
		d(y,y') 
		\coloneqq \sup_{ \substack{ v^+,v^- : \Delta(\Omega) \to \R
		\\ \text{continuous convex} \\ \text{s.t. $\abs{v^+-v^-} \leq 1$} } } 
		\abs*{ \int_{\Delta(\Omega)} \bigl( v^+ - v^- \bigr) \dd (y-y') } .
	\end{equation*}
	($d$ is in fact a metric on $\mathcal{Y}$.)
	Let $(y_n)_{n \in \N}$ be a sequence in $\mathcal{C}$ that converges to some $y \in \mathcal{C}$ in the relative topology on $\mathcal{C}$ inherited from the order topology on $\mathcal{Y}$; we will show that $d(y_n,y)$ vanishes as $n \to \infty$.

	Let
	$B_\eps
	\coloneqq \left\{ y' \in \mathcal{Y} : d(y,y') < \eps \right\}$
	denote the open $d$-ball of radius $\eps>0$ around $y$.
	Call $I \subseteq \mathcal{Y}$ an \emph{open order interval} iff either
	(1) $I = \{ y' \in \mathcal{Y} : y' < y^+ \}$ for some $y^+ \in \mathcal{Y}$, or
	(2) $I = \{ y' \in \mathcal{Y} : y^- < y' \}$ for some $y^- \in \mathcal{Y}$, or
	(3) $I = \{ y' \in \mathcal{Y} : y^- < y' < y^+ \}$ for some $y^- < y^+$ in $\mathcal{Y}$.
	Open order intervals are obviously open in the order topology on $\mathcal{Y}$.

	It suffices to show that for every $\eps > 0$, there is an open order interval $I_\eps \subseteq \mathcal{Y}$ such that $y \in I_\eps \subseteq B_\eps$.
	For then given any $\eps>0$, we know that $y_n$ lies in $I_\eps \intersect \mathcal{C} \subseteq B_\eps$ for all sufficiently large $n \in \N$ because (in the relative topology on $\mathcal{C}$) $I_\eps \intersect \mathcal{C}$ is an open set containing $y$ and $y_n \to y$.
	And this clearly implies that $d(y_n,y)$ vanishes as $n \to \infty$.

	So fix an $\eps > 0$; we will construct an open order interval $I \subseteq \mathcal{Y}$ such that $y \in I \subseteq B_\eps$.
	There are three cases.

	\emph{Case 1: $y' < y$ for no $y' \in \mathcal{Y}$.}
	Let $y^{++} \in \mathcal{Y}$ be such that $y < y^{++}$.
	Define
	\begin{equation*}
		y^+ \coloneqq (1-\eps/2) y + (\eps/2) y^{++} 
		\in \mathcal{Y}
		\quad \text{and} \quad
		I \coloneqq \bigl\{ y' \in \mathcal{Y} : 
		y' < y^+ \bigr\} .
	\end{equation*}
	We have $y < y^+$ and thus $y \in I$ since
	\begin{equation*}
		\int_{\Delta(\Omega)} v \dd\bigl( y^+ - y \bigr)
		= \frac{\eps}{2} \int_{\Delta(\Omega)} v \dd\bigl( y^{++} - y \bigr) 
	\end{equation*}
	is weakly (strictly) positive for every (some) continuous and convex $v : \Delta(\Omega) \to \R$ by $y < y^{++}$.
	To establish that $I \subseteq B_\eps$, it suffices to show that $d\bigl( y, y^+ \bigr) < \eps$, and this holds because
	\begin{equation*}
		d\bigl( y, y^+ \bigr)
		= \frac{\eps}{2} \sup_{ \substack{ v^+,v^- : \Delta(\Omega) \to \R
		\\ \text{continuous convex} \\ \text{s.t. $\abs{v^+-v^-} \leq 1$} } } 
		\abs*{ \int_{\Delta(\Omega)} \bigl( v^+ - v^- \bigr) \dd (y-y') }
		\leq \frac{\eps}{2}
		< \eps .
	\end{equation*}

	\emph{Case 2: $y < y'$ for no $y' \in \mathcal{Y}$.}
	This case is analogous to the first: choose a $y^{--} \in \mathcal{Y}$ such that $y^{--} < y$, and let
	\begin{equation*}
		y^- \coloneqq (1-\eps/2) y + (\eps/2) y^{--}
		\quad \text{and} \quad
		I \coloneqq \left\{ y' \in \mathcal{Y} : 
		y^- < y' \right\} .
	\end{equation*}
	The same arguments as in Case 1 yield $y \in I \subseteq B_\eps$.
	
	\emph{Case 3: $y' < y < y''$ for some $y',y'' \in \mathcal{Y}$.}
	Define $y^+$ as in Case 1 and $y^-$ as in Case 2, and let
	$I \coloneqq \bigl\{ y' \in \mathcal{Y} : 
	y^- < y' < y^+ \bigr\}$.
	We have $y \in I \subseteq B_\eps$ by the same arguments as in Cases 1 and 2.
\end{proof}

\begin{lemma}
	\label{lemma:diff-convex_approx}
	For any continuous function $c : \Delta(\Omega) \to \R$ and any $\eps>0$, there are continuous convex $w^+,w^- : \Delta(\Omega) \to \R$ such that $w \coloneqq w^+ - w^-$ satisfies
	$\sup_{ \mu \in \Delta(\Omega) } \abs*{ c(\mu) - w(\mu) } < \eps$.
\end{lemma}

\begin{proof}
	Write $\mathcal{W}$ for the space of functions $\Delta(\Omega) \to \R$ that can be written as the difference of continuous convex functions.
	Since the sum of convex functions is convex, $\mathcal{W}$ is a vector space.
	It is furthermore closed under pointwise multiplication \parencite[p. 708]{Hartman1959}, and thus an algebra.
	Clearly $\mathcal{W}$ contains the constant functions, and it separates points in the sense that for any distinct $\mu,\mu' \in \Delta(\Omega)$ there is a $w \in \mathcal{W}$ with $w(\mu) \neq w(\mu')$.
	It follows by the Stone--Weierstrass theorem%
		\footnote{See e.g. \textcite[Theorem 4.45]{Folland1999}.}
	that $\mathcal{W}$ is dense in the space of continuous functions $\Delta(\Omega) \to \R$ when the latter has the sup metric.
\end{proof}

With the lemmata in hand, we can verify the continuity hypothesis.

\begin{proposition}
	\label{proposition:pref_assns_ex}
	Consider the setting in §\ref{sec:app:info}.
	Let $\mathcal{C} \subseteq \mathcal{Y}$ be a chain, and equip it with the relative topology inherited from the order topology on $\mathcal{Y}$.
	Then $f$ is (jointly) continuous on $\mathcal{C} \times \R \times [0,1]$.
\end{proposition}

\begin{proof}
	Fix a chain $\mathcal{C} \subseteq \mathcal{Y}$,
	and equip it with the relative topology on $\mathcal{C}$ induced by the order topology on $\mathcal{Y}$.
	Define $h : \mathcal{C} \times [0,1] \to \R$
	by $h(y,t) \coloneqq \int_{\Delta(\Omega)} V(\mu,t) y(\dd \mu)$,
	so that $f(y,p,t) = g( h(y,t), p )$.
	Since $g$ is jointly continuous,
	we need only show that $h$ is jointly continuous.

	It suffices to prove that $h(\cdot,0)$ is continuous
	and that $\{ h_2(\cdot,t) \}_{t \in [0,1]}$ is equi-continuous.%
		\footnote{A detail: equi-continuity is a property of functions on a \emph{uniformisable} topological space.
		To see that $\mathcal{C}$ is uniformisable, we need only convince ourselves that the relative topology on $\mathcal{C}$ inherited from the order topology on $\mathcal{Y}$ is completely regular.
		This topology is obviously finer than the order topology on $\mathcal{C}$, so it suffices to show that the latter is completely regular.
		And that is (a consequence of) a standard result; see e.g. \textcite{Cater2006}.}
	To see why, take $(y,t)$ and $(y',t')$ in $\mathcal{C} \times [0,1]$ with (wlog) $t \leq t'$, and apply Lebesgue's fundamental theorem of calculus to obtain
	\begin{multline*}
		\abs*{ h(y',t') - h(y,t) }
		= \abs*{ h(y',0) + \int_0^{t'} h_2(y',s) \dd s
		- h(y,0) - \int_0^t h_2(y,s) \dd s }
		\\
		\leq \abs*{ h(y',0) - h(y,0) }
		+ \int_0^t \abs*{ h_2(y',s) - h_2(y,s) } \dd s
		+ \int_t^{t'} \abs*{ h_2(y',s) } \dd s .
	\end{multline*}
	Given continuity of $h(\cdot,0)$ (equi-continuity of $\{ h_2(\cdot,s) \}_{s \in [0,1]}$), the first term (second term) can be made arbitrarily small by taking $y$ and $y'$ sufficiently close (formally, choosing $y'$ in a neighbourhood of $y$ that is small in the sense of set inclusion).
	By boundedness of $h_2$, the third term can similarly be made small by choosing $t$ and $t'$ close.

	So take a sequence $(y_n)_{n \in \N}$ in $\mathcal{C}$ converging to some $y \in \mathcal{C}$; we must show that
	\begin{equation*}
		\abs*{ h(y_n,0) - h(y,0) }
		\quad \text{and} \quad
		\sup_{t \in [0,1]} \abs*{ h_2(y_n,t) - h_2(y,t) }
	\end{equation*}
	both vanish as $n \to \infty$.
	The former is easy: since $V(\cdot,0)$ is continuous (hence bounded) and convex, we have
	\begin{multline*}
		\abs*{ h(y_n,0) - h(y,0) } 
		= \abs*{ \int_{\Delta(\Omega)} V(\cdot,0) \dd (y_n-y) }
		\\
		\leq
		\left( \sup_{\mu \in \Delta(\Omega)} \abs*{ V(\mu,0) } \right)
		\times \sup_{ \substack{ v : \Delta(\Omega) \to [-1,1] 
		\\ \text{continuous convex} } } 
		\abs*{ \int_{\Delta(\Omega)} v \dd (y_n-y) }
	\end{multline*}
	for every $n \in \N$, and the right-hand side vanishes as $n \to \infty$ by \Cref{corollary:pref_assns_ex_conv}.

	For the latter, fix an $\eps>0$; we seek an $N \in \N$ such that
	\begin{equation*}
		\abs*{ h_2(y_n,t) - h_2(y,t) } < \eps
		\quad \text{for all $t \in [0,1]$ and $n \geq N$.}
	\end{equation*}
	For each $t \in [0,1]$, since $V_2(\cdot,t)$ is continuous, \Cref{lemma:diff-convex_approx} permits us to choose continuous and convex functions $w_t^+,w_t^- : \Delta(\Omega) \to \R$ such that $w_t \coloneqq w_t^+ - w_t^-$ is uniformly $\eps/3$-close to $V_2(\cdot,t)$.
	Write $K$ for the constant bounding $V_2$, and observe that $\{ w_t \}_{t \in [0,1]}$ is uniformly bounded by $K' \coloneqq K+\eps/3$.
	By \Cref{lemma:pref_assns_ex_conv}, there is an $N \in \N$ such that
	\begin{equation*}
		\sup_{ \substack{ v^+,v^- : \Delta(\Omega) \to \R
		\\ \text{continuous convex} \\ \text{s.t. $\abs{v^+-v^-} \leq 1$} } } 
		\abs*{ \int_{\Delta(\Omega)} \bigl( v^+ - v^- \bigr) \dd (y_n-y) }
		< \eps/3K'
		\quad \text{for all $n \geq N$,}
	\end{equation*}
	and thus
	\begin{equation*}
		\sup_{t \in [0,1]} 
		\abs*{ \int_{\Delta(\Omega)} 
		w_t \dd (y_n-y) }
		\\
		\leq K' \times \eps/3K'
		= \eps/3 
		\quad \text{for $n \geq N$.}
	\end{equation*}
	Hence for every $t \in [0,1]$ and $n \geq N$, we have
	\begin{align*}
		\abs*{ h_2(y_n,t) - h_2(y,t) }
		&= \abs*{ \int_{\Delta(\Omega)} V_2(\cdot,t) \dd (y_n-y) }
		\\
		&\leq 
		\abs*{ \int_{\Delta(\Omega)} w_t \dd (y_n-y) }
		+ \abs*{ \int_{\Delta(\Omega)} [ V_2(\cdot,t) - w_t ] \dd (y_n-y) }
		\\
		&\leq \abs*{ \int_{\Delta(\Omega)} w_t \dd (y_n-y) }
		+ 2 \sup_{\mu \in \Delta(\Omega)} \abs*{ V_2(\mu,t) - w_t(\mu) }
		\\
		&\leq \eps/3 + 2\eps/3
		= \eps ,
	\end{align*}
	as desired.
\end{proof}

\end{appendices}



\printbibliography[heading=bibintoc]


\end{document}